\documentclass[onecolumn, draftclsnofoot]{IEEEtran}

\usepackage{color}
\usepackage{graphicx}
\newtheorem{theorem}{Theorem}[section]
\newtheorem{lemma}[theorem]{Lemma}

\newcommand{\argmin}{\operatornamewithlimits{argmin}}

\ifCLASSINFOpdf
\else
\fi
\usepackage[cmex10]{amsmath}
\usepackage{algorithm, algorithmic}

\usepackage{array}

\usepackage[tight,footnotesize]{subfigure}
\usepackage{url}

\begin{document}

 \IEEEoverridecommandlockouts
\title{Spectrum Optimization in Multi-User Multi-Carrier Systems 
with Iterative Convex and Nonconvex Approximation Methods}

\author{\IEEEauthorblockN{Paschalis~Tsiaflakis$^*$\thanks{Paschalis
Tsiaflakis is a postdoctoral fellow funded by the Research
Foundation - Flanders (FWO). This text presents
results of the Belgian Program on Interuniversity Attraction
Poles initiated by the Belgian Federal Science Policy Office: IUAP P7/
‘Dynamical systems, control and optimization’ (DYSCO) 2012-2017, and IUAP
BESTCOM 2012-2017. Part of this research was conducted while the 
first author was supported by the Francqui Foundation as a postdoctoral intercommunity collaborator. The scientific responsibility rests with
its authors. A part of this paper has been presented at the IEEE International
Conference on Communications (ICC) \cite{confTsiaflakis2012}.}\thanks{This work has been submitted for possible publication. Copyright may be transferred without notice, after which this version may no longer be accessible.} and
Fran\c{c}ois~Glineur$^\dag$\\}
\IEEEauthorblockA{$^*$Department of Electrical Engineering (ESAT-SCD), KU Leuven, B-3000 Leuven,
Belgium\\$^\dag$Center for Operations Research and Econometrics,
Universit\'e catholique de Louvain, B-1348 Louvain-la-Neuve, Belgium\\
$^\dag$Information and Communication Technologies, Electronics and
Applied Mathematics Institute, Universit\'e catholique de Louvain, B-1348
Louvain-la-Neuve, Belgium\\ e-mail:
paschalis.tsiaflakis@esat.kuleuven.be,
francois.glineur@uclouvain.be}
}
%

%


\maketitle

\vspace{-1.7cm}
\begin{abstract}
Several practical multi-user multi-carrier communication systems are 
characterized by a multi-carrier interference channel system model where 
the interference is treated as noise. For these systems, spectrum 
optimization is a promising means to mitigate interference. This however 
corresponds to a challenging nonconvex optimization problem. Existing 
iterative convex approximation (ICA) methods consist in solving a series 
of improving convex approximations and are typically implemented in a 
per-user iterative approach. However they do not take this typical 
iterative implementation into account in their design. This paper 
proposes a novel class of iterative approximation methods that focuses 
explicitly on the per-user iterative implementation, which allows to 
relax the problem significantly, dropping joint convexity and even 
convexity requirements for the approximations. A systematic design 
framework is proposed to construct instances of this novel class, where 
several new iterative approximation methods are developed with improved 
per-user convex and nonconvex approximations that are both tighter and simpler to solve (in 
closed-form). As a result, these novel methods display a much faster 
convergence speed and require a significantly lower computational cost. 
Furthermore, a majority of the proposed methods can tackle the issue of 
getting stuck in bad locally optimal solutions, and hence improve solution 
quality compared to existing ICA methods. 

\end{abstract}



\begin{center} \textbf{ Index terms} -- Interference channel, spectrum optimization, interference mitigation, iterative approximation, multi-user multi-carrier \end{center}
%
\IEEEpeerreviewmaketitle

\section{Introduction}

The delivery of reliable and high-speed communication services to
a large set of wireless or wireline end-users has become a true challenge. 
This is due to the scarcity of the available frequency
bandwidth resources, the enormous growth of data traffic with stringent
quality-of-service (QoS) requirements, and the increasing number of
users. As a result many modern communication systems are forced to follow
a setup in which multiple users share a common frequency bandwidth that
is reused in time and frequency so as to improve the spectral efficiency.
However, in these systems \emph{interference} among the users impacts the
transmission of each user and can result in a significant performance
degradation such as decreased data rates and reduced reliability. In
fact, interference is identified as one of the major bottlenecks in
multi-user transmission \cite{Gesbert2010}. 

The development and optimization of multi-user interference mitigation techniques is
thus an important consideration. These techniques adopt an approach of
coordinating the transmission of multiple users taking the impact of
interference explicitly into account. In particular, interference
mitigation via \emph{spectrum optimization} is recognized as a very
powerful technique. It consists of the joint coordination of the
transmit spectrum, i.e., transmit powers over all frequency carriers, of
the interfering users so that the spectral efficiency is improved. This
can lead to a maximization of the data rates for a given
transmit power budget \cite{Cendrillon2006}, a minimization of the transmit powers for mininum 
target data rates \cite{energyDSLnordstrom,Monteiro2009}, some fair trade-off between these objectives \cite{Tsiaflakis2011},
or even a maximization of the signal to noise ratio (SNR) margin \cite{Wurtz2012,Monteiro2010}.
We want to highlight here that our target system model considers
spectrum optimization only, i.e., without any joint signal coordination
at the transmitter or the receiver side. From an information-theoretic point of view, our 
considered system corresponds to a \emph{multi-carrier interference channel}.
Here different users try to communicate their separate information over an
interference coupled channel. The capacity region of the interference channel
under all regimes is still an open problem. For specific subclasses of the
interference channels the capacity regions are obtained \cite{Han1981,Sato1981,Sung2008,Liu2008,Chong2008}, or outer
regions are derived \cite{Kramer2004,Etkin2008,Shang2009}. For practical systems where the couplings
are typically weaker than the direct channel, treating interference as additive
white Gaussian noise (AWGN) has been a very important practical communication strategy
in operational networks. This is a model that has been widely considered 
in literature. More specifically, spectrum optimization or power control (for multi-carrier and
single-carrier transmission, respectively) has been considered for
different systems: power control for code division multiple
access (CDMA) in cellular networks \cite{Chiang2007,Chiang2008}, power
control for cognitive wireless networks \cite{Haykin2005}, power control
for femtocell networks \cite{powerFemtocell,Femtocell2012}, general power
control \cite{CheeWeiTan2011, Weeraddana2011, Jing2012, Zhao2009}, spectrum
optimization  for multi-user multi-channel cellular relay networks
\cite{Ren2010}, spectrum optimization for heterogeneous wireless access
networks \cite{Son2011}, and cross-layer power and spectrum optimization
\cite{Weeraddana2011b,Tsiaflakis2008,Tsiaflakis2012}.  

In this paper we will focus on the general case of spectrum optimization
for multi-user multi-carrier systems where the interference is treated as 
AWGN. One important example is the digital
subscriber line (DSL) system where multiple users transmit over
twisted pair copper lines that are bundled in large cable binders. 
Each user adopts discrete multitone (DMT) modulation, a multicarrier
transmission technique where the frequency band is divided into
independent subcarriers, also referred to as tones. At the high
frequencies used by DSL DMT technology, electromagnetic
coupling among the different twisted pair copper lines occurs within the same
cable binder, which is also referred to as crosstalk interference or just
\emph{crosstalk}. Tackling this type of interference using spectrum
optimization is also referred to as dynamic spectrum management (DSM)
level-2, spectrum balancing or spectrum coordination in the DSL literature.
Note that although we will follow this line of spectrum optimization
theory developed for DSL wireline communications, the spectrum
optimization methods developed in this paper are similarly applicable for
general interference-limited multi-user systems that follow
an interference channel system model.

Many spectrum optimization methods have been proposed in the
literature to solve the corresponding \emph{spectrum optimization
problem}. These methods range from fully autonomous\footnote{Fully
autonomous spectrum optimization methods are methods in which each user chooses
its transmit powers autonomously, based on locally available information
only.} \cite{ASB,dsb,Leung10}, and distributed\footnote{Distributed spectrum optimization
methods are methods in which each user chooses its transmit
powers based on locally available information as well as information
obtained from other users through limited message-passing.} \cite{dsb,
MIW,scalejournal,Moraes2010,Wang2012}, to centralized methods\footnote{Centralized 
spectrum optimization methods are methods where the transmit powers of all users are
determined in a centralized location such as the spectrum management
center (SMC), where one has access to full knowledge of the channel
environment.} \cite{PBnB,tsiaflakis_bbosb,ISB_Raphael,dual_journal,Wolkerstorfer2012, Luo2009}. In particular, the approach
of \emph{iterative convex approximation} (ICA) has been recognized to be
very efficient,  as exemplified by the CA-DSB \cite{dsb,Tsiaflakis2007} and SCALE
\cite{scalejournal} algorithms. These algorithms consist in solving a series of improving
convex approximations to the original nonconvex problem, until convergence to a
locally optimal solution or a stationary point. They are furthermore
characterized by a low computational cost, high speed of convergence,
and potential distributed and asynchronous implementations. The
complexity of this ICA approach depends on two factors: (i) the type of
approximation, where a tighter approximation generally results in fewer
iterations until convergence, and (ii) the computational cost of solving
the corresponding approximation. Such ICA algorithms (for distributed or
centralized computation) are typically implemented in a per-user
iterative approach, as this is demonstrated to be highly efficient
\cite{Leung10,dsb,Huberman09,MIW,scalejournal,Moraes2010} for practical and implementational
reasons and it results in good performance. 
For instance, it is shown in \cite{dsb} that these iterative implementations can solve small- to medium-scale DSL scenarios, i.e., up to 6 users and 1000 tones, within a few seconds. However, for large-scale DSL scenarios, i.e. 10 to 50
users and 4000 tones, these iterative methods can take several minutes or even hours. Moreover, it is
shown in \cite{dsb} that these methods can sometimes get stuck in bad
locally optimal solutions.

In this paper we design a novel class of iterative approximation methods
that explicitly focus on the typical per-user iterative implementation.
It will be shown that in this setting, joint convexity and even convexity 
requirements on the approximations can be relaxed, resulting
in a much larger degree of freedom to design convergent iterative
algorithms. A design framework will be proposed that provides a design
strategy for developing improved (convex as well as nonconvex) approximations 
that are much \emph{tighter} and simpler to solve (i.e., \emph{in closed-form}) than those of existing state-of-the-art ICA methods.
With this design framework, \emph{ten novel methods} are developed that are provably
better than existing ICA methods in terms of faster convergence and
lower computational cost. Furthermore, it will be shown how some of the proposed
methods can tackle the issue of getting stuck in bad locally optimal solutions, 
and can thus even \emph{improve the final solution quality} (i.e., improve data rate performance).
Moreover, the flexibility of the proposed novel class allows the combination of
different types of approximations over different tones and users, which even
further improves the trade-off between computational cost and solution quality.
Both a reduction in the number of iterations required for convergence and an improvement in
solution quality will be demonstrated using a realistic DSL simulator. 

This text is organized as follows. Section~\ref{sec:systemmodel} briefly 
describes the multi-user multi-carrier system model and spectrum optimization. 
In Section~\ref{sec:ICA}, existing ICA spectrum optimization approaches are explained and
the typical per-user iterative implementation is highlighted. 
A general outline for our novel class of algorithms is described in Section~\ref{sec:generalalgorithm}, and the corresponding design framework is constructed in Section~\ref{sec:framework}. Using this framework, ten novel methods are constructed
in Section~\ref{sec:novelmethods}, which are analyzed and generalized in Section~\ref{sec:analysismethods}. 
Potential iterative fixed point update implementations
are considered and discussed in Section~\ref{sec:iterativefixedpoint}.
Improved solution quality of the proposed  methods based on nonconvex approximations
is discussed in Section~\ref{sec:nonconvexity}, and frequency and user
selective allocation of approximations is discussed in Section~\ref{sec:allocationapproximations}.
Finally, Section~\ref{sec:simulations} provides extensive simulation results for different realistic ADSL/ADSL2/VDSL scenarios and highlights the improved performance of the proposed novel class of algorithms with respect to that of state-of-the-art ICA methods in terms of convergence speed,  computational cost and solution quality.

\section{System Model and Spectrum Optimization}\label{sec:systemmodel}
We consider a multi-user multi-carrier system with a set
$\mathcal{N}=\{1,\ldots,N\}$ of $N$ users and a set
$\mathcal{K}=\{1,\ldots,K\}$ of $K$ tones (i.e., frequency
carriers). We consider synchronous transmission without intercarrier
interference. As we focus on spectrum optimization, we assume no signal
coordination at the transmitters or at the receivers. Furthermore, we make
the practical assumption that interference is treated as AWGN, which is for
instance a standard valid assumption in DSL communications when the number 
of interferers is large \cite{Kerpez93}.
Under these typical assumptions for spectrum coordination \cite{Cendrillon2006}, 
the \emph{spectrum optimization problem} (SO problem) can be formulated as
follows
\begin{equation}\label{eq:smp}
\begin{array}{lcl}
\mathrm{(SO~problem):}&{\displaystyle \min_{\mathbf{s}_k \in
\mathcal{S}_k, k \in \mathcal{K}}} & {\displaystyle \sum_{k \in
\mathcal{K}} f_k(\mathbf{s}_k)}\\
&\mathrm{subject~to} & {\displaystyle \sum_{k \in \mathcal{K}}
\mathbf{s}_k \leq \mathbf{P}^{\mathrm{tot}}},
\end{array} 
\end{equation}
with $\mathbf{s}_k=[s_k^1, \ldots, s_k^N]^T$ denoting transmit powers of
all users on tone $k$, the vector constant
$\mathbf{P}^{\mathrm{tot}}=[P^{1,\mathrm{tot}}, \ldots,
P^{N,\mathrm{tot}}]^T$ denoting total power budgets for all users,
and the set $\mathcal{S}_k=\{s_k^n, n \in \mathcal{N} \mid 0\leq s_k^n \leq
s_k^{n,\mathrm{mask}}, \forall n \in \mathcal{N}\}$ denoting the feasible
set on each tone $k$ with a maximum spectral transmit constraint
$s_k^{n,\mathrm{mask}}$ for user $n$. We also define the set 
$\mathcal{S}_k^n=\{s_k^n \mid 0\leq s_k^n \leq s_k^{n,\mathrm{mask}}\}$ which
will be used later in the text.

Each term $f_k(\mathbf{s}_k)$ in the objective function 
corresponds to the following nonconvex function
\begin{equation}\label{eq:smpobj}
f_k(\mathbf{s}_k)=-\sum_{n \in \mathcal{N}} w_n \log\left(
1+\frac{s_k^n}{{\displaystyle \sum_{m \neq n}} a_k^{n,m} s_k^m +
z_k^n}\right),
\end{equation}
with $w_n$ denoting the weight given to the achievable data rate of user $n$,
$a_k^{n,m}$ denoting the normalized channel gains from transmitter $m$ to
receiver $n$ on tone $k$, and $z_k^n$ denoting the normalized AWGN power
for line $n$ on tone $k$. The SNR gap \cite{UnderstandingDSL} is also included in the
normalized channel gains and noise power. The optimization problem (\ref{eq:smp}) with objective function (\ref{eq:smpobj}) corresponds to a standard weighted sum of 
achievable data rates maximization over multiple carriers and under a separate power constraint for each user.

We also define symbols $\mathrm{int}_k^n$ for the interference (plus noise) of user $n$
on tone $k$, and $\mathrm{rec}_k^n$ for the received signal of user $n$
on tone $k$, so as to facilitate notation in the remainder of the text:
\begin{eqnarray}
 & & \mathrm{int}_k^n = {\displaystyle \sum_{m \neq n} a_k^{n,m} s_k^m + z_k^n},\\
 & & \mathrm{rec}_k^n = {\displaystyle s_k^n + \sum_{m \neq n} a_k^{n,m} s_k^m + z_k^n = s_k^n + \mathrm{int}_k^n },
\end{eqnarray}
so that $f_k(\mathbf{s}_k)$ can be written as 
\begin{displaymath}
 f_k(\mathbf{s}_k)=-\sum_{n \in \mathcal{N}} w_n \log\left(
1+\frac{s_k^n}{\mathrm{int}_k^n}\right) = -\sum_{n \in \mathcal{N}} w_n \log\left(
\frac{\mathrm{rec}_k^n}{\mathrm{int}_k^n}\right).
\end{displaymath}
With the above definitions, we can also define the achievable data rates and bit loadings as
\begin{equation}
 R^n=\sum_{k \in \mathcal{K}} b_k^n=\sum_{k \in \mathcal{K}} \log \left( 1+\frac{s_k^n}{\mathrm{int}_k^n}\right),
\end{equation}
where $R^n$ and $b_k^n$ are the achievable data rate of user $n$ and the achievable
bit loading of user $n$ on tone $k$, respectively.

\section{Spectrum Optimization Through Iterative Convex
Approximation}\label{sec:ICA}
The ICA approach to tackle the SO problem consists in solving a series of
improving convex approximations. This procedure is given in
Algorithm~\ref{algo:ica}. In line 1 an initial approximation point $\mathbf{\tilde{s}}_k, k \in \mathcal{K},$
is chosen from the set $\{\mathbf{s}_k \in \mathcal{S}_k, k \in \mathcal{K}\}$. This 
initial point can be chosen based on a high
SNR heuristic \cite{scalejournal}, a zero transmit
power value \cite{dsb}, or even randomly \cite{Tsiaflakis2008,Tsiaflakis2012}. In line 3 the nonconvex
function $f_k(\mathbf{s}_k)$ is approximated by a convex function
$f_k^{\mathrm{cvx}}(\mathbf{s}_k; \mathbf{\tilde{s}}_k)$ 
built around the chosen approximation point $\mathbf{\tilde{s}}_k$. 
We want to highlight that the  convex approximation depends on the 
chosen approximation point $\mathbf{\tilde{s}}_k$. The resulting convex problem
is then solved in line~4 to obtain an improved solution, which is used
as a new approximation point in line 5. This is iterated in a loop (lines
2-6) until convergence to a stationary point. 
\begin{algorithm}
\caption{Iterative convex approximation (ICA) for SO problem
(\ref{eq:smp})}\label{algo:ica}
\begin{algorithmic}[1]
\STATE Choose an initial approximation point $\mathbf{\tilde{s}}_k \in
\mathcal{S}_k, k \in \mathcal{K}$
\WHILE{not converged to stationary point}
\STATE Approximate $f_k(\mathbf{s}_k)$ in problem (1) by convex function
$f_k^{\mathrm{cvx}}(\mathbf{s}_k; \mathbf{\tilde{s}}_k),  k \in
\mathcal{K}$
\STATE Solve corresponding convex problem to obtain
$\mathbf{s}_k^{\mathrm{cvx},*}, k \in \mathcal{K}$
\STATE Set $\mathbf{\tilde{s}}_k=\mathbf{s}_k^{\mathrm{cvx},*}, \forall k
\in \mathcal{K}$
\ENDWHILE
\end{algorithmic}
\end{algorithm}

The ICA procedure is guaranteed to converge to a locally optimal solution
or a stationary point of the SO problem under the following conditions
\cite{Marks78}:
\begin{eqnarray}
 &
f_k^{\mathrm{cvx}}(\mathbf{\tilde{s}}_k;\mathbf{\tilde{s}}_k)=f_k(\mathbf
{\tilde{s}}_k), \forall k \in \mathcal{K},& \label{cond:1}\\
  & \nabla f_k^{\mathrm{cvx}}(\mathbf{\tilde{s}}_k;\mathbf{\tilde{s}}_k)
= \nabla f_k(\mathbf{\tilde{s}}_k), \forall k \in \mathcal{K},&
\label{cond:2}\\
   & f_k^{\mathrm{cvx}}(\mathbf{s}_k;\mathbf{\tilde{s}}_k) \geq
f_k(\mathbf{{s}}_k), \forall \mathbf{s}_k \in \mathcal{S}_k, k \in
\mathcal{K},& \label{cond:3}
\end{eqnarray}
which impose that \begin{itemize}\item the approximation $f_k^{\mathrm{cvx}}(\mathbf{s}_k;
\mathbf{\tilde{s}}_k)$ coincides at first-order with the true objective
function $f_k(\mathbf{{s}}_k)$ around the approximation point
$\mathbf{\tilde{s}}_k, k \in \mathcal{K}$, with equations~\eqref{cond:1}--\eqref{cond:2}, 
\item  the approximation is an upper bound on the true objective function on the whole feasible set $\mathcal{S}_k$, via inequality~\eqref{cond:3}. \end{itemize} The CA-DSB and SCALE algorithms are two effective ICA procedures 
that use the following convex approximations:
\begin{eqnarray}
 & & f_k^{\mathrm{CADSB}}(\mathbf{s}_k,\mathbf{\tilde{s}}_k) = -\sum_{n \in \mathcal{N}} \left[ 
 w_n \log \left( \mathrm{rec}_k^n \right) \right] + \sum_{n \in \mathcal{N}}
 b_k^n(\mathbf{\tilde{s}}_k) s_k^n + c_k^n(\mathbf{\tilde{s}}_k)\\
 & & f_k^{\mathrm{SCALE}}(\mathbf{s}_k,\mathbf{\tilde{s}}_k) = - \sum_{n \in \mathcal{N}} 
 \left[ w_n \alpha_k^n(\mathbf{\tilde{s}}_k^n) \log \left( \frac{s_k^n}{\mathrm{int}_k^n}\right)\right] + c_k^n(\mathbf{\tilde{s}}_k)
\end{eqnarray}
in which parameters $b_k^n(\mathbf{\tilde{s}}_k)$, $\alpha_k^n(\mathbf{\tilde{s}_k})$ and
$c_k^n(\mathbf{\tilde{s}}_k)$ are constants that depend on the
approximation point $\mathbf{\tilde{s}}_k$,  and which are computed in closed-form
such that conditions (\ref{cond:1})-(\ref{cond:3}) are satisfied, see
\cite{dsb, scalejournal}. For instance, the following closed-form expressions
hold for $b_k^n(\mathbf{\tilde{s}}_k)$ and $\alpha_k^n(\mathbf{\tilde{s}}_k)$:
\begin{eqnarray}
& & b_k^n(\tilde{\mathbf{s}}_k)=\sum_{m \neq n} \frac{w_m
a_k^{m,n}}{{\displaystyle \mathrm{int}_k^m|_{\mathbf{s}_k^{-m}=\mathbf{\tilde{s}}_k^{-m}}}},\\
& & \alpha_k^n(\mathbf{\tilde{s}}_k) = \frac{\tilde{s}_k^n}{\mathrm{rec}_k^n|_{\mathbf{s}_k=\mathbf{\tilde{s}}_k}}.\label{eq:scaleconstant}
\end{eqnarray}
where $\mathrm{expression(\mathbf{x})}|_{\mathbf{x}=\mathbf{x}_0}$ means 
that the value of the variable $\mathbf{x}$ in the expression is fixed at $\mathbf{x}=\mathbf{x}_0$,
and where $\mathbf{s}_k^{-n}$ refers to the vector containing transmit powers on
tone $k$ of all users except for user $n$. As such, it can be seen that constructing the convex approximations,
i.e., line 3 of Algorithm~\ref{algo:ica}, is very simple, whereas line 4
corresponds to solving the convex problem and constitutes the main part
of the computational cost.

For the concrete implementation of ICA, one typically follows a \emph{per-user
iterative approach}, where each user iteratively computes its own transmit powers taking the fixed interference of the other users into account. This is done for several reasons: (i) it allows for distributed implementations, (ii) it allows the use of simple bisection searches to 
solve the per-user dual problem (as will be discussed later in this 
section), (iii) it allows to take the bound constraints into account with 
a simple projection operation, (iv) it doesn't require any stepsize tuning, 
and (v) most importantly, it has been demonstrated in the literature \cite{Huberman09,dsb} to have fast convergence and 
good performance for small- to medium-scale practical DSL scenarios. 
Because of these practically interesting and implementationally efficient
reasons, the per-user iterative implementation has been the preferred scheme for spectrum
optimization (in centralized \cite{dsb,Huberman09} as well as distributed 
implementations \cite{ASB,dsb,MIW,scalejournal,Moraes2010}).
The typical per-user iterative implementation is shown in Algorithm~{\ref{algo:userica}},
\algsetup{indent=0.4em}
\begin{algorithm}
\caption{Per-user iterative ICA approach}\label{algo:userica}
\begin{algorithmic}[1]
\STATE Choose an initial approximation point $\mathbf{\tilde{s}}_k \in
\mathcal{S}_k, k \in \mathcal{K}$
\FOR{outer iterations}
 \FOR{ user $n=1$ to $N$}
  \FOR{inner iterations}
   \STATE Update per-user convex approximation $f_k^{\mathrm{cvx}}(s_k^n;
\mathbf{\tilde{s}}_k^{-n}, \mathbf{\tilde{s}}_k)$ in $\tilde{s}_k^n, k \in \mathcal{K}$ to obtain (\ref{eq:usersmpcvx}) 
   \STATE Solve per-user convex problem (\ref{eq:usersmpcvx}) to obtain
$s_k^{\mathrm{cvx},n,*}, k \in \mathcal{K}$
   \STATE Set $\tilde{s}_k^n=s_k^{\mathrm{cvx},n,*}, k \in \mathcal{K}$
   \ENDFOR 
 \ENDFOR
\ENDFOR
\end{algorithmic}
\end{algorithm}
where line 2 corresponds to a fixed number of \emph{outer} iterations over all users (one can also use instead the stopping criterion used in line 2 of Algorithm~\ref{algo:ica}), line 3 is the
iteration over the users, and lines 4 to 8 focus on solving the following
per-user version of the SO problem for a given user $n$:
\begin{equation}\label{eq:usersmp}
\begin{array}{cl}
{\displaystyle \min_{s^n_k \in \mathcal{S}_k^n, k \in \mathcal{K}}} &
{\displaystyle \sum_{k \in \mathcal{K}} f_k(s^n_k; \mathbf{\tilde{s}}_k^{-n})}\\
\mathrm{subject~to} & {\displaystyle \sum_{k \in \mathcal{K}} s^n_k \leq
\mathbf{P}^{n,\mathrm{tot}}}.
\end{array} 
\end{equation}
More specifically, lines 4 to 8 come down to solving (\ref{eq:usersmp})
with an iterative approximation approach, performing a series of \emph{inner} iterations where each approximation (line 6 of Algorithm \ref{algo:userica}) corresponds to the following per-user 
problem 
\begin{equation}\label{eq:usersmpcvx}
\begin{array}{cl}
{\displaystyle \min_{s^n_k \in \mathcal{S}_k^n, k \in \mathcal{K}}} &
{\displaystyle \sum_{k \in \mathcal{K}} f_k^{\mathrm{cvx}}(s^n_k;
\mathbf{\tilde{s}}_k^{-n}, \mathbf{\tilde{s}}_k)}\\
\mathrm{subject~to} & {\displaystyle \sum_{k \in \mathcal{K}} s^n_k \leq
\mathbf{P}^{n,\mathrm{tot}}}.
\end{array} 
\end{equation}
We want to highlight here that the functions $f_k(s_k^n; \mathbf{s}_k^{-n})$ and $f_k^{\mathrm{cvx}}(s_k^n;\mathbf{s}_k^{-n},\mathbf{\tilde{s}}_k)$
are univariate (or one-dimensional) functions in $s_k^n$, because $\mathbf{s}_k^{-n}$
is fixed during the inner iterations of user $n$. The above convex approximated 
per-user problem (\ref{eq:usersmpcvx}) is then typically solved by focusing on the following dual
problem formulation \cite{dsb,scalejournal}, as this allows easier handling of the constraints:
\begin{equation}
{\displaystyle \max_{\lambda_n \geq 0}} ~~g( \lambda_n)
\label{eq:dualcvx} 
\end{equation}
\begin{equation}\label{eq:dualsub}
\mathrm{~~~~with~~} {\displaystyle g( \lambda_n) := -\lambda_n P^{n,\mathrm{tot}}+\sum_{k \in
\mathcal{K}}  \underbrace{\left[ \min_{s_k^n \in \mathcal{S}_k^n}
f_k^{\mathrm{cvx}}(s_k^n; \mathbf{\tilde{s}}_k^{-n},\mathbf{\tilde{s}}_k)+\lambda_n
s_k^n \right]}_{(\mathrm{SUB})}}.
\end{equation}
The dual problem (\ref{eq:dualcvx}) is a univariate problem in the
dual variable $\lambda_n$, which can be solved using a simple bisection
search, or (sub-)gradient update approaches \cite{tsiaflakis_bbosb}.
However, the evaluation of the objective function $g(\lambda_n)$
corresponds to an optimization problem on its own, i.e.,
(\ref{eq:dualsub}). This problem can be decomposed over tones for a given
$\lambda_n$, resulting in $K$ independent univariate subproblems, i.e., SUB in (\ref{eq:dualsub}),
which are convex subproblems as $f_k^{\mathrm{cvx}}(s_k^n;
\mathbf{\tilde{s}}_k^{-n},\mathbf{\tilde{s}}_k)$ is convex in $s_k^n$. This is also
referred to as \emph{dual decomposition}. For both
CA-DSB, i.e., $f_k^{\mathrm{cvx}}(s_k^n;
\mathbf{\tilde{s}}_k^{-n},\mathbf{\tilde{s}}_k)=f_k^{\mathrm{CADSB}}(s_k^n;
\mathbf{\tilde{s}}_k^{-n},\mathbf{\tilde{s}}_k)$, and SCALE, i.e.,
$f_k^{\mathrm{cvx}}(s_k^n;
\mathbf{\tilde{s}}_k^{-n},\mathbf{\tilde{s}}_k)=f_k^{\mathrm{SCALE}}(s_k^n;
\mathbf{\tilde{s}}_k^{-n},\mathbf{\tilde{s}}_k)$, it can be shown that each independent decomposed subproblem
in (\ref{eq:dualsub}) can be solved by computing the roots
of a polynomial of degree $N$, which can only be done in closed-form when $N\le4$,
i.e., with 4 users or less. Therefore iterative fixed
point updates were proposed in \cite{dsb,scalejournal} to solve the subproblems
in (\ref{eq:dualsub}), instead of computing closed-form solutions.


\begin{table*} 
\caption{Decomposition of 1D function for different methods, and max degree of
corresponding polynomial to solve subproblems in (\ref{eq:dualappsub}). Second naming scheme is described in Section~\ref{sec:analysismethods}. C and NC stand for convex and nonconvex.}
 \label{tab:approximations}
\begin{tabular}{ | c | l | l | c  | }
  \hline
 Methods & $f_k^1(s_k^n;\mathbf{s}_k^{-n},\mathbf{\tilde{s}}_k,\theta)$ & $f_k^2(s_k^n;\mathbf{s}_k^{-n},\mathbf{\tilde{s}}_k,\theta)$ & Degree\\
  \hline
  ORIG & $-w_n \log \left( 1+\frac{s_k^n}{{\displaystyle \mathrm{int}_k^n }}\right)$ & \slash & $2N-1$\\
  & $\qquad \displaystyle {-\sum_{m \in \mathcal{N}\backslash n}} \left[ w_m \log \left(1+\frac{s_k^m}{\mathrm{int}_k^m} \right)\right]$ & & \\
  \hline
  CADSB & $\displaystyle{-w_n \log \left(s_k^n + \mathrm{int}_k^n \right)}$ & $\displaystyle{+\sum_{n \in \mathcal{N}}} \left[ w_n \log \left(
\mathrm{int}_k^n \right) \right]$ &  $N$\\
  & $\qquad \displaystyle{-\sum_{m \in \mathcal{N}\backslash n}} \left[ w_m \log \left(s_k^m + \mathrm{int}_k^m \right) \right]$ & & \\
  \hline
  SCALE & $- w_n \alpha_k^n(\mathbf{\tilde{s}}_k)\log \left( \frac{s_k^n}{\mathrm{int}_k^n} \right) - c_k^n(\mathbf{\tilde{s}}_k)$ & 
$\displaystyle{\sum_{n \in \mathcal{N}}} \left[ w_n
\alpha_k^n(\mathbf{\tilde{s}}_k)\log \left( \frac{s_k^n}{\mathrm{int}_k^n} \right)\right]+ c_k^n(\mathbf{\tilde{s}}_k)$ & $N$ \\
& $\displaystyle{-\sum_{m \in \mathcal{N}\backslash n}} \left[ w_m \alpha_k^m(\mathbf{\tilde{s}}_k)\log \left( \frac{s_k^m}{\mathrm{int}_k^m} \right)\right]$ & \qquad \qquad \qquad \qquad \qquad $\displaystyle {-\sum_{n \in \mathcal{N}}} \left[ w_n \log \left(1+\frac{s_k^n}{\mathrm{int}_k^n} \right)\right]$ & \\
\hline
$\begin{matrix} \text{IASB1} \\ \text{(IA1, C)} \end{matrix}$  & $-w_n \log
\left( 1+\frac{s_k^n}{{\displaystyle \mathrm{int}_k^n }}\right)$ & 
$\displaystyle {-\sum_{m \in \mathcal{N}\backslash n}} \left[ w_m \log \left(1+\frac{s_k^m}{\mathrm{int}_k^m} \right)\right]$ & $1$\\
\hline
$\begin{matrix} \text{IASB2} \\ \text{(IA2-$L$, C/NC)} \end{matrix}$ & $-w_n \log
\left( 1+\frac{s_k^n}{{\displaystyle \mathrm{int}_k^n }}\right)-L_k^n(s_k^n-\tilde{s}_k^n)^2$ & 
$\displaystyle {-\sum_{m \in \mathcal{N}\backslash n}} \left[ w_m \log \left(1+\frac{s_k^m}{\mathrm{int}_k^m} \right)\right]+L_k^n(s_k^n-\tilde{s}_k^n)^2$ & $2$\\
\hline
$\begin{matrix} \text{IASB3} \\ \text{(IA3-$r$, NC)} \end{matrix}$  & $ \displaystyle{ - w_n \log \left( 1+\frac{s_k^n}{ \mathrm{int}_k^n}\right)
- w_q \log \left( 1+\frac{s_k^q}{ \mathrm{int}_k^q}\right)}$ & $\displaystyle{- \sum_{p \in \mathcal{N}\backslash \{n,q\}}} \left[ w_p
\log \left( 1+\frac{s_k^p}{ \mathrm{int}_k^p}\right)\right]$ & $3$\\
\hline
$\begin{matrix} \text{IASB4} \\ \text{(IA2-$\alpha$, NC)} \end{matrix}$ & $-w_n \log \left( 1+\frac{s_k^n}{\mathrm{int}_k^n}\right)$ & 
$+w_q \alpha_k^q(\mathbf{\tilde{s}}_k) \log
\left( \frac{s_k^q}{\mathrm{int}_k^q}\right) + c_k^n(\mathbf{\tilde{s}}_k) $ & $2$\\
& \qquad \qquad \qquad $-w_q \alpha_k^q(\mathbf{\tilde{s}}_k) \log \left( \frac{s_k^q}{\mathrm{int}_k^q}\right) - c_k^n(\mathbf{\tilde{s}}_k)$ & 
\qquad \qquad \qquad \qquad $\displaystyle {-\sum_{m \in \mathcal{N}\backslash n}} \left[ w_m \log \left(1+\frac{s_k^m}{\mathrm{int}_k^m} \right)\right]$ &\\
\hline
$\begin{matrix} \text{IASB5} \\ \text{(IA3-$\alpha^2$, NC)} \end{matrix}$ & $-w_n \log
\left( 1+\frac{s_k^n}{\mathrm{int}_k^n}\right)$ & ${\displaystyle \sum_{p \in \{q,t\}}} \left[ w_p
\alpha_k^p(\mathbf{\tilde{s}}_k) \log \left(\frac{s_k^p}{\mathrm{int}_k^p}\right)\right] + c_k^n(\mathbf{\tilde{s}}_k)$ & $3$\\
& $- {\displaystyle \sum_{p \in \{q,t\}}} \left[ w_p
\alpha_k^p(\mathbf{\tilde{s}}_k) \log \left(\frac{s_k^p}{\mathrm{int}_k^p}\right)\right] - c_k^n(\mathbf{\tilde{s}}_k)$& $\displaystyle {-\sum_{m \in \mathcal{N}\backslash n}} \left[ w_m \log \left(1+\frac{s_k^m}{\mathrm{int}_k^m} \right)\right]$ &\\
\hline
$\begin{matrix} \text{IASB6} \\ \text{(IA1-$\beta$, C)} \end{matrix}$ & $-w_n (1-\beta_k^n) \log
\left( 1+\frac{s_k^n}{{\displaystyle \mathrm{int}_k^n }}\right)$ & 
$-\beta_k^n w_n \log \left( 1+\frac{s_k^n}{{\displaystyle \mathrm{int}_k^n }}\right)
\displaystyle {-\sum_{m \in \mathcal{N}\backslash n}} \left[ w_m \log \left(1+\frac{s_k^m}{\mathrm{int}_k^m} \right)\right]$ & $1$\\
\hline
$\begin{matrix} \text{IASB7} \\ \text{(IA3-$\beta$r, NC)} \end{matrix}$ & $- w_n (1-\beta_k^n) \log \left( 1+\frac{s_k^n}{ \mathrm{int}_k^n}\right)- w_q
\log \left( 1+\frac{s_k^q}{ \mathrm{int}_k^q}\right)$ & $
-\beta_k^n w_n \log \left( 1+\frac{s_k^n}{{\displaystyle \mathrm{int}_k^n }}\right)\displaystyle{- \sum_{p \in \mathcal{N}\backslash \{n,q\}}} \left[ w_p
\log \left( 1+\frac{s_k^p}{ \mathrm{int}_k^p}\right)\right]$ & $3$\\
\hline
$\begin{matrix} \text{IASB8} \\ \text{(IA3-$\beta\alpha^2$, NC)} \end{matrix}$ & $-w_n (1-\beta_k^n) \log
\left( 1+\frac{s_k^n}{\mathrm{int}_k^n}\right)$ & $
-\beta_k^n w_n \log \left( 1+\frac{s_k^n}{{\displaystyle \mathrm{int}_k^n }}\right) \displaystyle{-\sum_{m \in \mathcal{N}\backslash n}} \left[ w_m \log \left(1+\frac{s_k^m}{\mathrm{int}_k^m} \right)\right]$ & $3$\\
& \qquad $- {\displaystyle \sum_{p \in \{q,t\}}} \left[ w_p
\alpha_k^p(\mathbf{\tilde{s}}_k) \log \left(\frac{s_k^p}{\mathrm{int}_k^p}\right)\right]- c_k^n(\mathbf{\tilde{s}}_k)$ & \qquad \qquad \qquad $+ {\displaystyle \sum_{p \in \{q,t\}}} \left[ w_p
\alpha_k^p(\mathbf{\tilde{s}}_k) \log \left(\frac{s_k^p}{\mathrm{int}_k^p}\right)\right] + c_k^n(\mathbf{\tilde{s}}_k)$ & \\
\hline
$\begin{matrix} \text{IASB9} \\ \text{(IA3-$\alpha$L, NC)} \end{matrix}$ & $-w_n \log
\left( 1+\frac{s_k^n}{\mathrm{int}_k^n}\right) -L_k^n(s_k^n-\tilde{s}_k^n)^2 $ & 
$+w_q \alpha_k^q(\mathbf{\tilde{s}}_k) \log
\left( \frac{s_k^q}{\mathrm{int}_k^q}\right) + c_k^n(\mathbf{\tilde{s}}_k) +L_k^n(s_k^n-\tilde{s}_k^n)^2$ & $3$\\
& \qquad \qquad \qquad $-w_q \alpha_k^q(\mathbf{\tilde{s}}_k) \log
\left( \frac{s_k^q}{\mathrm{int}_k^q}\right) - c_k^n(\mathbf{\tilde{s}}_k)$ & \qquad \qquad \qquad \qquad \qquad \qquad $\displaystyle {-\sum_{m \in \mathcal{N}\backslash n}} \left[ w_m \log \left(1+\frac{s_k^m}{\mathrm{int}_k^m} \right)\right]$ & \\
\hline
$\begin{matrix} \text{IASB10} \\ \text{(IA$N$-$\alpha^{N-1}$, NC)} \end{matrix}$ & $-w_n \log(1+\frac{s_k^n}{\mathrm{int}_k^n})$ & 
$\displaystyle{\sum_{m \in \mathcal{N}\backslash n}} \left[ w_m
\alpha_k^m(\mathbf{\tilde{s}}_k)\log \left( \frac{s_k^m}{\mathrm{int}_k^m} \right)\right]+ c_k^m(\mathbf{\tilde{s}}_k)$ & $N$\\
& $\displaystyle{-\sum_{m \in \mathcal{N}\backslash n}} \left[ w_m
\alpha_k^m(\mathbf{\tilde{s}}_k)\log \left( \frac{s_k^m}{\mathrm{int}_k^m} \right)\right] - c_k^m(\mathbf{\tilde{s}}_k)$ & 
\qquad \qquad \qquad \qquad \qquad $\displaystyle {-\sum_{m \in \mathcal{N}\backslash n}} \left[ w_m \log \left(1+\frac{s_k^m}{\mathrm{int}_k^m} \right)\right]$ & \\
\hline
\end{tabular}
\end{table*}

The above per-user iterative approach using dual decomposition is recognized as
being very effective in tackling the SO problem in the sense that it can
find locally optimal solutions to the SO problem with only small
computational cost. More specifically, it is shown in \cite{dsb} that the
SO problem up to dimension $KN=2000$ can be solved within a few seconds.

However, there are also some drawbacks when using existing ICA methods 
(CA-DSB and SCALE). First, the SO problem can have many locally
optimal solutions, depending on the considered scenario, where many
of those locally optimal solutions can correspond to a quite suboptimal
data rate performance, as demonstrated in \cite{dsb}. Existing ICA
methods feature no mechanism to tackle this issue and, depending on the
chosen initial point and the considered scenario, may converge to a
locally optimal solution with very seriously deteriorated data rate
performance. A second issue is that for large-scale scenarios, e.g., DSL scenarios with
10-100 users and more than 1000 tones, ICA methods may take several
minutes or even hours. Any improvement to reduce the execution time is
desirable for such large-scale scenarios.

\section{A Novel Class of Iterative Approximation Methods}\label{sec:iterativemethods}

Although existing ICA methods are typically implemented in a per-user
iterative fashion as shown in the previous section, their design does not take this iterative
implementation into account. We propose a novel class of algorithms 
whose design is tailored to this per-user iterative implementation.
We will show that this approach allows to relax the per-user 
problem to be solved at each iteration, 
resulting in a much larger design space and can lead to
significantly improved convergent iterative spectrum optimization
algorithms.

\subsection{General Algorithm for Novel Class}\label{sec:generalalgorithm}
When analyzing the typical iterative implementation of
Algorithm~\ref{algo:userica}, it can be observed that lines 5 and 6
approximate and solve the per-user problem (\ref{eq:usersmpcvx}) and thus, following a typical
dual decomposition approach, line~6 resorts to solving (possibly in parallel) $K$ independent univariate
approximations in the variable $s_k^n$, 
one for each tone $k$, as in (\ref{eq:dualsub}). The univariate approximations
can be considered as per-user per-tone approximations.

Existing ICA methods rely on approximations 
$f_k^{\mathrm{cvx}}(\mathbf{s}_k,\mathbf{\tilde{s}}_k)$ that are jointly convex in
$\mathbf{s}_k$. While this joint convexity property guarantees that each per-user problem (13) is convex, it is not a necessary condition (i.e., there are situations where per-user problems are convex while $f_k^{\mathrm{cvx}}(\mathbf{s}_k,\mathbf{\tilde{s}}_k)$ is not jointly convex). Moreover, although convexity of the per-user problems typically ensures that they can be efficiently minimized, it is not a strict requirement as some nonconvex problems still admit fast resolution techniques. Therefore we propose to relax this joint convexity restriction
on $f_k^{\mathrm{cvx}}(\mathbf{s}_k,\mathbf{\tilde{s}}_k)$ and even drop the requirement of being univariately convex in $s_k^n$.
This results in the novel class of per-user iterative methods of Algorithm~\ref{algo:novel}, where
the main difference is that they rely on a univariate approximating function
$f_k^{\mathrm{app}}(s_k^n;\mathbf{\tilde{s}}_k^{-n},\mathbf{\tilde{s}}_k,\theta)$
that is not necessarily jointly convex in $\mathbf{s}_k$ or even convex in
$s_k^n$. Here, $\theta$ stands for a tuning parameter that will be exploited
later in the text to further tighten the approximation. Note that our approach 
starts from the iterative per-user implementation as a hard constraint in which 
the unnecessary constraint of joint convexity
or even univariate convexity for the approximations is relaxed. We will show
that this provides a larger design freedom where approximations can be proposed that are 
tighter (i.e., have a smaller approximation error along the considered dimension of user $n$)
and for which the subproblems in (\ref{eq:dualsub}), can be solved more easily, i.e., even in closed-form.
As a result this novel class of algorithms leads to fewer inner iterations to converge
to the local optimum of that user's iteration, i.e., performs less iterations of lines 4 to 8,
and in addition requires a lower computational cost per iteration. We only impose that the approximating function
$f_k^{\mathrm{app}}(s_k^n;\mathbf{\tilde{s}}_k^{-n},\mathbf{\tilde{s}}_k,\theta)$ univariately 
satisfies the convergence conditions 
(\ref{cond:1}), (\ref{cond:2}) and (\ref{cond:3}) as follows:
\begin{eqnarray}
 &
f_k^{\mathrm{app}}(\tilde{s}_k^n;\mathbf{\tilde{s}}^{-n}_k,\mathbf{\tilde{s}}_k,
\theta)= f_k(\tilde{s}_k^n;\mathbf{\tilde{s}}_k^{-n}), \forall k \in \mathcal{K},& \label{cond_uni:1}\\
  & \nabla f_k^{\mathrm{app}}(\tilde{s}_k^n;\mathbf{\tilde{s}}^{-n}_k,\mathbf{\tilde{s}}_k,
\theta) = \nabla f_k(\tilde{s}_k^n;\mathbf{\tilde{s}}_k^{-n}), \forall k \in \mathcal{K},&
\label{cond_uni:2}\\
& f_k^{\mathrm{app}}(s_k^n;\mathbf{\tilde{s}}^{-n}_k,\mathbf{\tilde{s}}_k,
\theta) \geq
f_k(s_k^n;\mathbf{\tilde{s}}_k^{-n}), \forall s^n_k \in \mathcal{S}_k^n, k \in
\mathcal{K}.& \label{cond_uni:3}
\end{eqnarray}
The corresponding
per-user approximated problem becomes
\begin{equation}\label{eq:usersmpapp}
\begin{array}{cl}
{\displaystyle \min_{s^n_k \in \mathcal{S}_k^n, k \in \mathcal{K}}} &
{\displaystyle \sum_{k \in \mathcal{K}} f_k^{\mathrm{app}}(s_k^n;\mathbf{\tilde{s}}_k^{-n},\mathbf{\tilde{s}}_k,\theta)}\\
\mathrm{subject~to} & {\displaystyle \sum_{k \in \mathcal{K}} s^n_k \leq
\mathbf{P}^{n,\mathrm{tot}}},
\end{array} 
\end{equation}
and its dual decomposed formulation is as follows
\begin{equation}
{\displaystyle \max_{\lambda_n \geq 0}} ~~g^{\mathrm{app}}( \lambda_n)
\label{eq:dualapp} 
\end{equation}
\begin{equation}\label{eq:dualappsub}
\mathrm{~~~~with~~} {\displaystyle g^{\mathrm{app}}( \lambda_n) := -\lambda_n P^{n,\mathrm{tot}}+\sum_{k \in
\mathcal{K}}  \underbrace{\left[ \min_{s_k^n \in \mathcal{S}_k^n}
f_k^{\mathrm{app}}(s_k^n; \mathbf{\tilde{s}}_k^{-n},\mathbf{\tilde{s}}_k,\theta)+\lambda_n
s_k^n \right]}_{(\mathrm{SUB})}}.
\end{equation}

\algsetup{indent=0.4em}
\begin{algorithm}
\caption{Novel class of per-user iterative approximation
methods}\label{algo:novel}
\begin{algorithmic}[1]
\STATE Choose an initial approximation point $\mathbf{\tilde{s}}_k \in
\mathcal{S}_k, k \in \mathcal{K}$
\FOR{outer iterations}
 \FOR{ user $n=1$ to $N$}
  \FOR{inner iterations}
   \STATE Update per-user approximation $f_k^{\mathrm{app}}(s_k^n;\mathbf{\tilde{s}}_k^{-n},\mathbf{\tilde{s}}_k,\theta)$ in $\tilde{s}_k^n, k \in \mathcal{K}$
   \STATE Solve per-user dual problem (\ref{eq:dualapp}) with bisection and approximated subproblems (SUB) in (\ref{eq:dualappsub}), \\to obtain
$s_k^{\mathrm{app},n,*}, k \in \mathcal{K}$
   \STATE Set $\tilde{s}_k^n=s_k^{\mathrm{app},n,*}, k \in \mathcal{K}$
   \ENDFOR 
 \ENDFOR
\ENDFOR
\end{algorithmic}
\end{algorithm}

For this novel class of iterative methods, as given in Algorithm~\ref{algo:novel}, 
the only remaining task is to design improved convex or nonconvex
per-user approximations $f_k^{\mathrm{app}}$ (univariate in each tone) that are both tighter and easier to solve, i.e., in closed-form.
Note that for existing ICA methods the univariate approximations on each
tone cannot be solved in closed-form and one has to resort to an iterative fixed
point approach, where multiple iterations are required to converge to the solution
of the univariate approximation.


Finally, we describe convergence of Algorithm~\ref{algo:novel}. By enforcing 
conditions (\ref{cond_uni:1})-(\ref{cond_uni:3}), convergence to a locally 
optimal solution is guaranteed for Algorithm~\ref{algo:novel} under the condition 
that, when the approximation is built around a non-locally optimal point, the dual optimization step (line 6 of Algorithm~\ref{algo:novel}) results in a solution with a decreased objective function value $f_k^{\mathrm{app}}$. 
The optimal value of each subproblem approximation, i.e., (SUB) in (\ref{eq:dualappsub}), is always 
non-increasing by design. For the dual optimization step, although the 
approximations are allowed to be nonconvex, the theoretical duality gap between
(\ref{eq:usersmpapp}) and (\ref{eq:dualapp}) tends to zero because of the 
large number of independent tones \cite{dual_journal,DSMluo}, i.e., $K$ is very large. 
In our extensive simulation environment, when using the dual decomposition approach, 
we always observed convergence to primal feasible solutions, on each of more than a hundred of realistic large scale DSL scenarios.


\subsection{Design framework for improved approximations}\label{sec:framework}

To obtain a class of convergent algorithms with improved convergence speed
and reduced computational cost, we need a structured way to construct improved univariate
approximations $f_k^{\mathrm{app}}$ for Algorithm~\ref{algo:novel}. 
Therefore we propose a general \emph{design framework} that
can be followed to obtain such improved approximations.


For this, we start from the true 
(nonconvex) per-user objective $f_k(s_k^n;\mathbf{s}_k^{-n})$ as given in the first row of
Table~{\ref{tab:approximations}}, which is a univariate function in $s_k^n$ that needs to be approximated around $s_k^n=\tilde{s}_k^n$. In a first step we decompose this function as a sum of two 
functions as follows 
\begin{equation}\label{eq:decomposition}
 f_k(s_k^n;\mathbf{s}_k^{-n}) = f_k^1(s_k^n;\mathbf{s}_k^{-n},\mathbf{\tilde{s}}_k,\theta) + 
 f_k^2(s_k^n;\mathbf{s}_k^{-n},\mathbf{\tilde{s}}_k,\theta),
\end{equation} 
with 
\begin{eqnarray}
f_k^1(s_k^n;\mathbf{\tilde{s}}_k^{-n},\mathbf{\tilde{s}}_k,\theta)& \in & \mathcal{F}^1(\mathbf{\tilde{s}}_k) \label{eq:condf1}\\ 
f_k^2(s_k^n;\mathbf{\tilde{s}}_k^{-n},\mathbf{\tilde{s}}_k,\theta)& \in & \mathcal{F}^2(\mathbf{\tilde{s}}_k) \label{eq:condf2}
\end{eqnarray}
where the sets of functions $\mathcal{F}^1(\mathbf{\tilde{s}}_k)$ and $\mathcal{F}^2(\mathbf{\tilde{s}}_k)$ 
are defined as follows:
\begin{eqnarray}
\mathcal{F}^1 = \{ \mathrm{univariate~functions~whose~derivative~is~a~rational~function~of    ~degree~at~most~} 3, \mathrm{~as~in~(\ref{eq:rationalf})} \} & &\label{cond:f1} \end{eqnarray}
\begin{equation}\label{eq:rationalf}
 \frac{p_1 (s_k^n)^3+p_2 (s_k^n)^2+p_3 s_k^n+p_4}{
p_5 (s_k^n)^3+p_6 (s_k^n)^2+p_7 s_k^n+p_8 }, \mathrm{~~~defined~by~constants~} p_1 \mathrm{~to~} p_8,
\end{equation}
\begin{eqnarray}
\mathcal{F}^2 = \{ \mathrm{univariate~functions~that~are~upperbounded~by~their~tangent~linear~approximation~in~} s_k^n = \tilde{s}_k^n 
\}. & & \label{cond:f2}
\end{eqnarray}

Set $\mathcal{F}^1$ corresponds to a broad set of functions, where the only constraint is that the derivative
corresponds to a rational function of maximal degree equal to 3, as made explicit in (\ref{eq:rationalf}). 
Set $\mathcal{F}^2$ clearly includes all concave functions, but there exist many non-concave functions that also satisfy
this property, and thus the set of functions $\mathcal{F}^2$ is much larger than the set of concave functions
only. Although there exist an infinite number of decompositions of type (\ref{eq:decomposition}), this
general decomposition provides a design approach for generating improved univariate
approximations, as will become clear later in this section.
Once a specific decomposition satisfying
(\ref{eq:decomposition}) is chosen, the second step is to construct an univariate approximating 
function by linearizing $f_k^2(s_k^n;\mathbf{s}_k^{-n},\mathbf{\tilde{s}}_k,\theta)$ as follows:
\begin{equation}\label{eq:fapp}
\begin{array}{rl}
 f_k^{\mathrm{app}}(s_k^n;\mathbf{s}_k^{-n},\mathbf{\tilde{s}}_k,\theta) & =  f_k^1(s_k^n;\mathbf{s}_k^{-n},\mathbf{\tilde{s}}_k,\theta) + \mathrm{LIN}(f_k^2(s_k^n;s_k^{-n},\mathbf{\tilde{s}}_k,\theta),\\
& = f_k^1(s_k^n;\mathbf{s}_k^{-n},\mathbf{\tilde{s}}_k,\theta) + d_k^n(\theta) s_k^n + e_k^n(\theta),
\end{array}
\end{equation}
where $\mathrm{LIN}(\cdot)$ denotes the first-order linear approximation operator around the
current approximation point $s_k^n=\tilde{s}_k^n$, corresponding to the following explicit choice of constants
$d_k^n(\theta)$ and $e_k^n(\theta)$:
\begin{equation}
\begin{array}{rl}
 d_k^n(\theta) = & \frac{\partial}{\partial s_k^n} (f_k^2(\tilde{s}_k^n;\mathbf{\tilde{s}}_k^{-n},\mathbf{\tilde{s}}_k,\theta)), \\
 e_k^n(\theta) = & f_k^2(\tilde{s}_k^n;\mathbf{\tilde{s}}_k^{-n},\mathbf{\tilde{s}}_k,\theta)-d_k^n(\theta) \tilde{s}_k^n.
 \end{array}
\end{equation}
The following theorem characterizes the properties of the constructed approximating function $f_k^{\mathrm{app}}(s_k^n;s_k^{-n},\mathbf{\tilde{s}}_k,\theta)$.
\begin{theorem}\label{theorem:approx}
 Univariate approximating function $f_k^{\mathrm{app}}(s_k^n;\mathbf{\tilde{s}}_k^{-n},\mathbf{\tilde{s}}_k,\theta)$ satisfies 
 convergence conditions (\ref{cond_uni:1}), (\ref{cond_uni:2}) and (\ref{cond_uni:3}) along the considered dimension $s_k^n$, and its corresponding
subproblems SUB in (\ref{eq:dualappsub}), when following a dual decomposition approach,
can be solved in closed-form with a maximal computational cost equal to
that of solving one cubic equation and performing five function evaluations of (\ref{eq:fapp}).
\end{theorem}
\begin{proof}
 \begin{itemize}
  \item As $f_k^2(s_k^n;\mathbf{\tilde{s}}_k^{-n},\mathbf{\tilde{s}}_k,\theta) \in \mathcal{F}^2(\mathbf{\tilde{s}}_k)$, 
  $f_k^2(s_k^n;\mathbf{\tilde{s}}_k^{-n},\mathbf{\tilde{s}}_k,\theta)$ is upperbounded by $\mathrm{LIN}(f_k^2(s_k^n;\mathbf{\tilde{s}}_k^{-n},\mathbf{\tilde{s}}_k,\theta))$.
  As a result, $f_k^{\mathrm{app}}(s_k^n;\mathbf{\tilde{s}}_k^{-n},\mathbf{\tilde{s}}_k,\theta) \geq 
  f_k(s_k^n;\mathbf{\tilde{s}}_k^{-n}), \forall s_k^n \in \mathcal{S}_k^n$, and thus (\ref{cond_uni:3}) is satisfied.
  \item By definition of $d_k^n(\theta)$, $f_k^2(s_k^n;\mathbf{\tilde{s}}_k^{-n},\mathbf{\tilde{s}}_k,\theta)$ and $d_k^n(\theta) s_k^n + e_k^n(\theta)$ have the 
  same derivative in $s_k^n=\tilde{s}_k^n$, i.e., $\frac{\partial}{\partial s_k^n} f_k^2(\tilde{s}_k^n;\mathbf{\tilde{s}}_k^{-n},\mathbf{\tilde{s}}_k,\theta)
  = d_k^n(\theta)$, and thus (\ref{cond_uni:2}) is satisfied
  along the considered dimension $s_k^n$.
  \item By definition of $e_k^n(\theta)$, $f_k^2(s_k^n;\mathbf{\tilde{s}}_k^{-n},\mathbf{\tilde{s}}_k,\theta)$ and $d_k^n(\theta) s_k^n + e_k^n(\theta)$ have the same
  function value in $s_k^n=\tilde{s}_k^n$, and thus $f_k(\tilde{s}_k^n;\mathbf{\tilde{s}}_k^{-n})= f_k^{\mathrm{app}}(\tilde{s}_k^n;\mathbf{\tilde{s}}_k^{-n},\mathbf{\tilde{s}}_k,\theta)$, i.e., (\ref{cond_uni:1}) is satisfied.
  \item Each subproblem SUB in (\ref{eq:dualappsub}) corresponds to the following optimization problem:
  \begin{equation}\label{eq:fappsub}
   \min_{s_k^n \in \mathcal{S}_k^n} \left[ f_k^{\mathrm{app}}(s_k^n;\mathbf{\tilde{s}}_k^{-n},\mathbf{\tilde{s}}_k,\theta)+\lambda_n s_k^n \right].
  \end{equation}
  Note that the feasible set $\mathcal{S}_k^n$ corresponds to a simple interval (i.e., lower and 
  upper bound constraints). Although the objective function of
  (\ref{eq:fappsub}) is not necessarily convex in $s_k^n$, the minimization
  is over a scalar variable only, and the optimal value can be easily found.
  More specifically, for the considered one dimensional case, the optimal
  solution must be located either where the derivative is zero or at one of the
  two boundary points. For our concrete case, the first order
  optimality condition, i.e, zero derivative, is given as follows:
  \begin{equation}\label{eq:stationaritycondition}
   \nabla_{s_k^n}f_k^{\mathrm{app}}(s_k^n;\mathbf{\tilde{s}}_k^{-n},\mathbf{\tilde{s}}_k,\theta) = 0  \Rightarrow \nabla_{s_k^n}f_k^{1}(s_k^n;\mathbf{\tilde{s}}_k^{-n},\mathbf{\tilde{s}}_k,\theta)+\lambda_n+d_k^n(\theta)=0.
  \end{equation}
  With the help of (\ref{eq:rationalf}), which provides a closed for the derivative of $f^1_k$, this can be reformulated as
  \begin{equation}\label{eq:polynomial}
   (p_1-p_5 (\lambda_n+d_k^n(\theta)))(s_k^n)^3+(p_2-p_6 (\lambda_n+d_k^n(\theta)))(s_k^n)^2+(p_3-p_7 (\lambda_n+d_k^n(\theta)))s_k^n+(p_4-p_8 (\lambda_n+d_k^n(\theta)))=0.
  \end{equation}
  As this corresponds to a cubic equation in one variable ($s_k^n$), it can be solved in closed-form, resulting in at most three roots $r_k^n(1), r_k^n(2)$ and $r_k^n(3)$.
  By comparing the value of the objective function of (\ref{eq:fappsub}) at those of these three roots that are feasible (i.e., belong to $\mathcal{S}_k^n$), as well as at the boundary
  points, i.e., $s_k^n=0$ and $s_k^n=s_k^{n,\mathrm{mask}}$, the optimal value of (\ref{eq:fappsub}) can be obtained.
  Formally, this can be written as
  \begin{equation}\label{eq:userapp}
   s_k^{n,*}:=\argmin_{x \in \mathcal{X}} f_k^{\mathrm{app}}(s_k^n;\mathbf{\tilde{s}}_k^{-n},\mathbf{\tilde{s}}_k,\theta)+\lambda_n x,
  \end{equation}
  with 
  \begin{equation}
    \mathcal{X}=\left\{0,s_k^{n,\mathrm{mask}},[r_k^n(1)]_0^{s_k^{n,\mathrm{mask}}},[r_k^n(2)]_0^{s_k^{n,\mathrm{mask}}},[r_k^n(3)]_0^{s_k^{n,\mathrm{mask}}}\right\}
  \end{equation}
(where $[x]_a^b$ denotes $\max(\min(x,b),a)$, i.e. the projection of $x$ on interval $[a\ b]$). Under our decomposition assumption, each problem (\ref{eq:fappsub}) can therefore be exactly minimized by solving one cubic equation and performing at most five function evaluations. 
 \end{itemize}
\end{proof}

Thus, following this design framework guarantees that the approximating function $f_k^{\mathrm{app}}(s_k^n;\mathbf{\tilde{s}}_k^{-n},\mathbf{\tilde{s}}_k,\theta)$
satisfies the necessary properties so as to obtain a convergent algorithm when using the sequential updates of Algorithm~\ref{algo:novel}.

The proposed decomposition in (\ref{eq:decomposition}) is chosen so 
as to explicitly decouple the design for low computational cost and the
design for small approximation error. More specifically, function $f_k^2(s_k^n;\mathbf{\tilde{s}}_k^{-n},\mathbf{\tilde{s}}_k,\theta)$
does not contribute to the computational cost of solving the approximated
problem (\ref{eq:fappsub}). This is because its influence vanishes to
a constant in the optimality condition after the linearization step in (\ref{eq:fapp}). The computational cost is fully determined
by function $f_k^1(s_k^n;\mathbf{\tilde{s}}_k^{-n},\mathbf{\tilde{s}}_k,\theta)$, more specifically, by the maximal degree of the polynomials
in the numerator and denominator of (\ref{eq:rationalf}). Note that the
rational function can have a degree smaller than 3. In this case, the
resulting polynomial (\ref{eq:polynomial}) corresponds to a quadratic or
even a linear equation for which it is even simpler to solve the
subproblem. Note that this is much simpler than existing ICA methods, 
for which the closed-form solution requires computing the roots of a
polynomial of degree $N$, as mentioned in Section~\ref{sec:ICA}.
Concrete examples will be given in Section \ref{sec:novelmethods}.

The approximation error then again is fully determined by $f_k^2(s_k^n;\mathbf{\tilde{s}}_k^{-n},\mathbf{\tilde{s}}_k,\theta)$
and the corresponding linearization operation. The design guideline is
to choose $f_k^2(s_k^n;\mathbf{\tilde{s}}_k^{-n},\mathbf{\tilde{s}}_k,\theta)$ such that it resembles a linear function as much as possible, so that the approximation error is minimized and the approximation becomes
tighter. Furthermore, the additional parameter
$\theta$ can be tuned to further improve the approximation, while satisfying
the constraints of the decomposition into functions $f_k^1(s_k^n;\mathbf{\tilde{s}}_k^{-n},\mathbf{\tilde{s}}_k,\theta)$ and $f_k^2(s_k^n;\mathbf{\tilde{s}}_k^{-n},\mathbf{\tilde{s}}_k,\theta)$.
More specifically, we propose to choose $\theta$ such that the
absolute value of the second derivative of $f_k^2(s_k^n;\mathbf{\tilde{s}}_k^{-n},\mathbf{\tilde{s}}_k,\theta)$ is minimized in the interval 
$\mathcal{S}_k^n$, 
\begin{equation}\label{eq:tuning}
\theta^* := \argmin_{\theta} \left| \left( \frac{\partial^2}{\partial
(s_k^n)^2} \left(f_k^2(s_k^n;\mathbf{\tilde{s}}_k^{-n},\mathbf{\tilde{s}}_k,\theta)\right) \right)\right|.
\end{equation}
The rationale for this choice is that the smaller (in absolute value) the second derivative is, the closer to linear the function is.

\subsection{Novel Methods}\label{sec:novelmethods}
In this section the design framework proposed in Section~\ref{sec:framework} 
will be used to develop different improved univariate approximations $f_k^{\mathrm{app}}$ .
These approximations can be used in Algorithm~\ref{algo:novel} 
to obtain novel iterative spectrum optimization 
algorithms with faster convergence, reduced computational cost and even
improved achievable data rate performance.

\subsubsection{Iterative Approximation Spectrum Balancing 1
(IASB1)}\label{sec:IASB1}

For our first method, we propose the IASB1 decomposition of Table~\ref{tab:approximations},
where the second and third columns correspond to the two decomposed terms
$f_k^{1,\mathrm{IASB1}}$ and $f_k^{2,\mathrm{IASB1}}$, respectively.
Following the design framework of Section~\ref{sec:framework}, we only have to show that these functions satisfy conditions (\ref{eq:condf1}) and (\ref{eq:condf2}), respectively.
Condition (\ref{eq:condf2}) for $f_k^2$ holds as $f_k^2$ is concave in $s_k^n$: indeed its second derivative is negative for all $s_k^n \in \mathcal{S}_k^n$. Secondly, it can be observed that the derivative
of $f_k^1$ corresponds to a rational function of degree 1 and thus (\ref{eq:condf1}) is
also satisfied. Consequently, applying the linearization step of the design framework
leads to the following approximation 
\begin{equation}
f_k^{\mathrm{app}}(s_k^n;\mathbf{\tilde{s}}_k^{-n},\mathbf{\tilde{s}}_k,\theta)=f_k^{\mathrm{IASB1}}(s_k^n;\mathbf{\tilde{s}}_k^{-n},\mathbf{\tilde{s}}_k,\theta)=f_k^{\mathrm{1,IASB1}}(s_k^n;\mathbf{\tilde{s}}_k^{-n},\mathbf{\tilde{s}}_k,\theta)+d_k^n(\theta) 
s_k^n + e_k^n(\theta),
\end{equation}
that satisfies all the necessary properties, such as the convergence
conditions (\ref{cond_uni:1}), (\ref{cond_uni:2}) and (\ref{cond_uni:3}), as proven by Theorem~\ref{theorem:approx}.
The resulting approximation $f_k^{\mathrm{IASB1}}(s_k^n;\mathbf{\tilde{s}}_k^{-n},\mathbf{\tilde{s}}_k,\theta)$ is convex in $s_k^n$ as $f_k^1(s_k^n;\mathbf{\tilde{s}}_k^{-n},\mathbf{\tilde{s}}_k,\theta)$ is convex in $s_k^n$. 
The following lemma shows that its tightness is improved compared to CA-DSB.
\begin{lemma}\label{lemma:iasb1}
Approximation $f_k^{\mathrm{IASB1}}(s_k^n;\mathbf{\tilde{s}}_k^{-n},\mathbf{\tilde{s}}_k,\theta)$ is tighter than approximation
$f_k^{\mathrm{CADSB}}(s_k^n;\mathbf{\tilde{s}}_k^{-n},\mathbf{\tilde{s}}_k)$.
\end{lemma}
\begin{IEEEproof}
Within the proposed design framework, the approximation error only depends on 
function $f_k^2(s_k^n;\mathbf{\tilde{s}}_k^{-n},\mathbf{\tilde{s}}_k,\theta)$. Functions $f_k^{2,\mathrm{IASB1}}(s_k^n;\mathbf{\tilde{s}}_k^{-n},\mathbf{\tilde{s}}_k,\theta)$ and $f_k^{2,\mathrm{CA-DSB}}(s_k^n;\mathbf{\tilde{s}}_k^{-n},\mathbf{\tilde{s}}_k)$ 
are both concave. Since first-order conditions coincide
in $s_k^n=\tilde{s}_k^n$, proving an inequality between second derivatives everywhere on the interval will imply that one approximation is tighter than the other one.
Indeed, the smaller the absolute value of the second derivative in the interval, the smaller the 
approximation error, as also highlighted in Section~\ref{sec:framework} in
(\ref{eq:tuning}). 
The corresponding second derivatives are given as follows
\begin{eqnarray}
& & \frac{\partial^2}{\partial (s_k^n)^2}\left(f_k^{2,\mathrm{CA-DSB}}(s_k^n;\mathbf{\tilde{s}}_k^{-n},\mathbf{\tilde{s}}_k,\theta)\right)=-\sum_{m \in \mathcal{N}\backslash n} \frac{w_m (a_k^{m,n})^2}{(\mathrm{int}_k^m)^2},\\
& & \frac{\partial^2}{\partial (s_k^n)^2}\left(f_k^{2,\mathrm{IASB1}}(s_k^n;\mathbf{\tilde{s}}_k^{-n},\mathbf{\tilde{s}}_k,\theta)\right)=-\sum_{m \in \mathcal{N}\backslash n} \frac{w_m (a_k^{m,n})^2}{(\mathrm{int}_k^m)^2}
\underbrace{\left(
1-\left(\frac{\mathrm{int}_k^m}{\mathrm{rec}_k^m}\right)^2\right)}_{(A)}.\label{eq:secondderiasb1}
\end{eqnarray}
As $0 \leq \mathrm{int}_k^m \leq \mathrm{rec}_k^m, \forall s_k^n \in \mathcal{S}_k^n$, factor (A) in (\ref{eq:secondderiasb1})
is non-negative and smaller than $1$. Therefore we have
$0\geq \frac{\partial}{\partial s_k^n}(f_k^{2,\mathrm{IASB1}}(s_k^n;\mathbf{\tilde{s}}_k^{-n},\mathbf{\tilde{s}}_k,\theta)) \geq 
\frac{\partial}{\partial s_k^n}(f_k^{2,\mathrm{CA-DSB}}(s_k^n;\mathbf{\tilde{s}}_k^{-n},\mathbf{\tilde{s}}_k), \forall s_k^n \in \mathcal{S}_k^n$, and thus IASB1
corresponds to a tighter approximation.
\end{IEEEproof}

Besides improved tightness for IASB1, we obtain a significant reduction in
computational cost for solving the subproblems SUB in (\ref{eq:dualappsub}). More specifically, it can be derived
that the corresponding polynomial (\ref{eq:polynomial}) is a simple linear equation, which can be solved 
in closed-form as follows:
\begin{equation}\label{eq:updateIASB1}
s_k^n = \left[ \frac{w_n}{\lambda_n + d_k^n(\theta)} -
\left( \sum_{m \neq n} a_k^{n,m} s_k^m + z_k^n \right)
\right]^{s_k^{n,\mathrm{mask}}}_0
\end{equation}
where $[x]_a^b$ means $\max(\min(x,b),a)$. 
The proposed per-user approximation $f_k^{\mathrm{IASB1}}(s_k^n;\mathbf{\tilde{s}}_k^{-n},\mathbf{\tilde{s}}_k,\theta)$ for IASB1 (\ref{eq:userapp}) 
is thus both tighter than that of CA-DSB, and much easier to solve (in
closed-form) than CA-DSB and SCALE. The computational costs for the subproblems
considered in this paper are summarized in the last column of Table~\ref{tab:approximations}. The
concrete improvement in number of required approximations to converge 
and computational cost observed in practice will be demonstrated in Section \ref{sec:simulations}
for realistic DSL scenarios.

Finally, we want to note that the transmit power formula (\ref{eq:updateIASB1}) of
IASB1 corresponds to that of the DSB algorithm proposed in \cite{dsb}.
However, both update formulas are derived in a fundamentally different
way, where IASB1 gives some important additional insights. More
specifically, it shows that (\ref{eq:updateIASB1}) is the solution of a
convex problem satisfying conditions (\ref{cond_uni:1})-(\ref{cond_uni:3}), which
proves that this per-user iterative update is non-increasing and
converges to a (univariate) local minimum under the sequential iterative
updates, properties which were not known previously for DSB.

\subsubsection{Iterative Approximation Spectrum Balancing 2 (IASB2)}
The per-user approximation of IASB1 $f_k^{\mathrm{IASB1}}(s_k^n;\mathbf{\tilde{s}}_k^{-n},\mathbf{\tilde{s}}_k,\theta)$ can be further 
improved using the decomposition $f_k^n(s_k^n;\mathbf{\tilde{s}}_k^{-n})=f_k^{1,IASB2}(s_k^n;\mathbf{\tilde{s}}_k^{-n},\mathbf{\tilde{s}}_k,L_k^n)+f_k^{2,IASB2}(s_k^n;\mathbf{\tilde{s}}_k^{-n},\mathbf{\tilde{s}}_k,L_k^n)$, 
whose terms are given in the second and third columns of Table~{\ref{tab:approximations}}, respectively.
Compared to IASB1, it introduces an additional quadratic term with
constant leading coefficient $L_k^n$. The reason for adding this term to $f_k^{2,\mathrm{IASB1}}$ is to
decrease the approximation error, by making the absolute value of its
second derivative smaller in the interval $\mathcal{S}_k^n$ compared to that of
IASB1, i.e., $\left|\frac{\partial^2}{\partial (s_k^n)^2}(f_k^{2,IASB2}(s_k^n;\mathbf{\tilde{s}}_k^{-n},\mathbf{\tilde{s}}_k,L_k^n))\right| \leq \left|\frac{\partial^2}{\partial (s_k^n)^2}(f_k^{2,IASB1}(s_k^n;\mathbf{\tilde{s}}_k^{-n},\mathbf{\tilde{s}}_k,\theta))\right|$.
A positive value for $L_k^n$ makes $f_k^{\mathrm{IASB2}}(s_k^n;\mathbf{\tilde{s}}_k^{-n},\mathbf{\tilde{s}}_k,L_k^n)$  a lower
approximation of $f_k^{\mathrm{IASB1}}(s_k^n;\mathbf{\tilde{s}}_k^{-n},\mathbf{\tilde{s}}_k,\theta)$, and thus a better approximation. Note that the value of $L_k^n$
cannot be chosen too large to ensure that $f_k^{\mathrm{2,IASB2}}(s_k^n;\mathbf{\tilde{s}}_k^{-n},\mathbf{\tilde{s}}_k,L_k^n)$ remains
concave in the interval $\mathcal{S}_k^n$, which is required so that $f_k^{2,\mathrm{IASB2}}(s_k^n;\mathbf{s}_k^{-n},\mathbf{\tilde{s}}_k,L_k^n) \in \mathcal{F}^2(\mathbf{\tilde{s}}_k)$. We
propose the following constant positive value for $L_k^n$:
\begin{equation}\label{eq:Lvalue}
L_k^n = {\displaystyle \sum_{m \neq n}} \frac{w_m
\hat{s}_k^m (a_k^{m,n})^2 \left( \hat{s}_k^m+2\sum_{p \neq m} a_k^{m,p}
\hat{s}_k^p+2z_k^m \right)/2} {((\hat{s}_k^m+\sum_{p \neq m} a_k^{m,p}
\hat{s}_k^p+z_k^m)(\sum_{p \neq m} a_k^{m,p} \hat{s}_k^p+z_k^m))^2},
\end{equation}
with $\hat{\mathbf{s}}_k = (s_k^{n,\mathrm{mask}}, \mathbf{\tilde{s}}_k^{-n})$, which
ensures that $f_k^{2,IASB2}(s_k^n;\mathbf{\tilde{s}}_k^{-n},\mathbf{\tilde{s}}_k,L_k^n)$ has minimum curvature.  
The value (\ref{eq:Lvalue}) is obtained by taking the second derivative of $f_k^{2,IASB2}(s_k^n;\mathbf{\tilde{s}}_k^{-n},\mathbf{\tilde{s}}_k,L_k^n)$ and 
choosing $L_k^n$ such that it becomes zero in one point only, namely $s_k^n=s_k^{n,\mathrm{mask}}$, and negative otherwise. One can show that it
corresponds to the closed-form solution of problem (\ref{eq:tuning}), with
$\theta=L_k^n$. After linearization, the resulting univariate approximation $f_k^{IASB2}(s_k^n;\mathbf{\tilde{s}}_k^{-n},\mathbf{\tilde{s}}_k,L_k^n)$
is obtained, which is not necessarily a convex function. One can also derive 
a convex version by adding an additional constraint on tuning parameter $L_k^n$ so that
$f_k^{1,IASB2}(s_k^n;\mathbf{\tilde{s}}_k^{-n},\mathbf{\tilde{s}}_k,L_k^n)$ remains convex. This results in the following value (also proposed in \cite{confTsiaflakis2012})
\begin{equation}\label{eq:Lvalue2}
\begin{array}{rl}
& L_k^n =  \min \left(x,y\right) \\
\mathrm{with} & x = \frac{w_n}{2\left(\hat{s}_k^n+\sum_{m \neq n}
a_k^{n,m} \hat{s}_k^m+z_k^n \right)^2}\\
		       & y = {\displaystyle \sum_{m \neq n}} \frac{w_m
\hat{s}_k^m (a_k^{m,n})^2 \left( \hat{s}_k^m+2\sum_{p \neq m} a_k^{m,p}
\hat{s}_k^p+2z_k^m \right)/2} {((\hat{s}_k^m+\sum_{p \neq m} a_k^{m,p}
\hat{s}_k^p+z_k^m)(\sum_{p \neq m} a_k^{m,p} \hat{s}_k^p+z_k^m))^2},
\end{array}
\end{equation}
with $\hat{\mathbf{s}}_k = (s_k^{n,\mathrm{mask}}, \mathbf{\tilde{s}}_k^{-n})$.

The improved tightness of $f_k^{IASB2}(s_k^n;\mathbf{\tilde{s}}_k^{-n},\mathbf{\tilde{s}}_k,L_k^n)$ with both values for $L_k^n$ (\ref{eq:Lvalue}) and (\ref{eq:Lvalue2})
is stated by the following lemma:
\begin{lemma}
Per-user approximation $f_k^{\mathrm{IASB2}}$ is tighter than
approximations $f_k^{\mathrm{IASB1}}$ and $f_k^{\mathrm{CA-DSB}}$.
\end{lemma}
\begin{IEEEproof}
Proof is trivial since $L_k^n \geq 0$ implies that $0\geq \frac{\partial^2}{\partial (s_k^n)^2}(f_k^{2,\mathrm{IASB2}})
= \frac{\partial^2}{\partial (s_k^n)^2}(f_k^{2,\mathrm{IASB1}}) + 2L_k^n \geq 
\frac{\partial^2}{\partial (s_k^n)^2}(f_k^{2,\mathrm{IASB1}}), \forall s_k^n \in \mathcal{S}_k^n,$ and considering Lemma~\ref{lemma:iasb1}.
\end{IEEEproof}

In terms of computational cost to solve the subproblems, it can be derived
that polynomial equation  (\ref{eq:polynomial}) for approximation $f_k^{\mathrm{IASB2}}$
is quadratic. The approximated problem can thus be solved 
in closed-form, by checking the values of two roots only in addition 
to the boundary points, in contrast to three roots as in (\ref{eq:userapp}).

In summary, approximation $ f_k^{\mathrm{IASB2}}$ is tighter than $
f_k^{\mathrm{IASB1}}$, but requires a slightly higher computational cost
to solve the corresponding subproblems in closed-form.

\subsubsection{Iterative Approximation Spectrum Balancing 3
(IASB3)}\label{sec:IASB3}
Our next approximation also improves that of $f_k^{\mathrm{IASB1}}$. 
It starts from the decomposition IASB3 as given in Table~\ref{tab:approximations}.
It fits the design framework as $f_k^{1,\mathrm{IASB3}}$ has a derivative
that corresponds to a rational function of degree 3, i.e., $f_k^{1,\mathrm{IASB3}} \in \mathcal{F}^1$,
and $f_k^{2,IASB3}$ corresponds to a concave function in $s_k^n$ leading to $f_k^{2,IASB3} \in \mathcal{F}^2$.
The resulting approximating function $f_k^{\mathrm{IASB3}}$, after linearization
of $f_k^{2,\mathrm{IASB3}}$, can be nonconvex in $s_k^n$. A concrete illustration
of the nonconvexity will be given in Section~\ref{sec:quality} and Figure~\ref{fig:approximationsADSLJournal}.
The improved tightness of the IASB3 approximation is stated in the following lemma:
\begin{lemma}
Per-user approximation $f_k^{\mathrm{IASB3}}$ is tighter than
approximations $f_k^{\mathrm{IASB1}}$ and $f_k^{\mathrm{CA-DSB}}$.
\end{lemma}
\begin{IEEEproof}
Proof is trivial since $\left| \frac{\partial^2}{\partial (s_k^n)^2}
\left( f_k^{2,\mathrm{IASB3}}\right)\right| \leq \left| \frac{\partial^2}{\partial (s_k^n)^2}
\left( f_k^{2,\mathrm{IASB1}}\right)\right|, \forall s_k^n \in \mathcal{S}_k^n,$ as $f_k^{2,\mathrm{IASB3}}$
corresponds to $f_k^{2,\mathrm{IASB1}}$ with one less concave term.
\end{IEEEproof}
As the derivative of $f_k^{1,\mathrm{IASB1}}$ corresponds to a rational
function of degree 3, solving a cubic equation is required to find the
solution of each subproblem SUB in (\ref{eq:dualappsub}) in closed-form, as in (\ref{eq:stationaritycondition}) with $f_k^{\mathrm{app}}=
f_k^{\mathrm{IASB3}}=f_k^{\mathrm{1,IASB3}}+d_k^n(\theta) s_k^n + e_k^n(\theta)$.

We would like to highlight here that the IASB3 method corresponds to the
ASB method of \cite{ASB} if $d_k^n(\theta)$ and $e_k^n(\theta)$ of $f_k^{\mathrm{IASB3}}$ are set to zero. The addition
of $d_k^n(\theta)$ and $e_k^n(\theta)$ however ensures that IASB3 converges to a per-user
local optimum, as opposed to ASB. One huge advantage of IASB3 compared to
CA-DSB, SCALE and IASB1 is that the nonconvex nature of
$f_k^{\mathrm{IASB3}}$ allows to get out of bad locally optimal
solutions, as will be demonstrated in Section~\ref{sec:nonconvexity},
which can result in improvements in the solution data rate performance. 

Finally we would like to mention that a good choice of index $q$ is relevant
and influences the solution quality and convergence speed. The `reference
line' heuristics for ASB provided in \cite{ASB,Leung10} can be used for
this purpose. A further analysis of this choice is a subject of further study.

\subsubsection{Iterative Approximation Spectrum Balancing 4
(IASB4)}\label{sec:IASB4}

This method starts from decomposition IASB4 in Table~\ref{tab:approximations},
where the constants $\alpha_k^q(\mathbf{\tilde{s}}_k)$ and $c_k^q(\mathbf{\tilde{s}}_k)$ are chosen so that the following
inequality holds
\begin{equation}\label{eq:iasb4app}
\begin{array}{c}
-w_q\log(1+\frac{s_k^q}{\mathrm{int}_k^q}) \leq -w_q \alpha_k^q(\mathbf{\tilde{s}}_k)
\log(\frac{s_k^q}{\mathrm{int}_k^q})-c_k^q(\mathbf{\tilde{s}}_k),
\end{array}
\end{equation}
with equality holding at one point, namely for $s_k^n=\tilde{s}_k^n$. Note that
the same inequality is also exploited by the SCALE algorithm. The above
constants can be determined easily in closed-form, e.g., the closed-form
expression for $\alpha_k^n(\mathbf{\tilde{s}}_k)$ is given in (\ref{eq:scaleconstant}).
The derivative of $f_k^{1,\mathrm{IASB4}}$ corresponds to a rational function
of degree 2, i.e., $f_k^{1,\mathrm{IASB4}} \in \mathcal{F}^1$, and the corresponding polynomial (\ref{eq:polynomial})
to solve each subproblem (\ref{eq:dualappsub}) is a simple quadratic which can be solved
in closed-form. The following lemma proves that $f_k^{2,\mathrm{IASB4}}$
can be upperbounded by a tangent linear function, and thus $f_k^{2,\mathrm{IASB4}} \in \mathcal{F}^2$:
\begin{lemma}
 $f_k^{2,\mathrm{IASB4}}$ is upperbounded by the (unique) linear function tangent at the approximation point $s_k^n=\tilde{s}_k^n$.
\end{lemma}
\begin{IEEEproof}
First we rewrite $f_k^{2,\mathrm{IASB4}}$:
\begin{equation}\label{eq:IASB4_parts}
 f_k^{2,\mathrm{IASB4}} =  \underbrace{+w_q \alpha_k^q(\mathbf{\tilde{s}}_k) \log
\left( \frac{s_k^q}{\mathrm{int}_k^q}\right) + c_k^n(\mathbf{\tilde{s}}_k) 
- w_q \log \left( 1+\frac{s_k^q}{\mathrm{int}_k^q}\right)}_{(C)} \underbrace{\displaystyle {-\sum_{m \in \mathcal{N}\backslash \{n,q\}}} \left[ w_m \log \left(1+\frac{s_k^m}{\mathrm{int}_k^m} \right)\right]}_{(D)}.
\end{equation}
Part (D) of (\ref{eq:IASB4_parts}) is concave in $s_k^n$ and can thus be upperbounded by a linear tangent function in $s_k^n=\tilde{s}_k^n$.
Based on the inequality (\ref{eq:iasb4app}) it can be seen that part (C) is strictly negative, except in the point $s_k^n=\tilde{s}_k^n$
where it has zero value. Thus, part (C) can be linearly upperbounded by the constant function with constant value zero.
As a result, the overall function $f_k^{2,\mathrm{IASB4}}$ can be upperbounded by a linear tangent in $s_k^n = \tilde{s}_k^n$, 
and thus $f_k^{2,\mathrm{IASB4}} \in \mathcal{F}^2$.
\end{IEEEproof}
We emphasize here that part (C) of (\ref{eq:IASB4_parts}) is not
concave. This justifies our definition of function set $\mathcal{F}^2$, instead of simply restricting ourselves to concave functions.

The following lemma characterizes the approximation tightness:
\begin{lemma}\label{lemma:iasb4app}
 Approximation $f_k^{\mathrm{IASB4}}$ is less tight than approximation
 $f_k^{\mathrm{IASB3}}$ and more tight than approximation $f_k^{\mathrm{IASB1}}$.
\end{lemma}
\begin{IEEEproof}
The approximation error is determined by the relative tightness of $f_k^{2,\mathrm{IASB4}}$
versus that of $f_k^{2,\mathrm{IASB1}}$ and $f_k^{2,\mathrm{IASB3}}$.
The difference between $f_k^{2,\mathrm{IASB4}}$ and $f_k^{2,\mathrm{IASB3}}$ corresponds
to part (C) of (\ref{eq:IASB4_parts}) which is negative for $s_k^n \neq \tilde{s}_k^n$.
This means that $f_k^{\mathrm{2,IASB4}}$ is a lower bound for $f_k^{\mathrm{2,IASB3}}$,
and thus has a larger approximation error after linearization.
The difference between $f_k^{2,\mathrm{IASB4}}$ and $f_k^{2,\mathrm{IASB1}}$ corresponds
to the first term of (\ref{eq:IASB4_parts}), which is convex in $s_k^n$. 
The addition of this additional convex term in $f_k^{2,\mathrm{IASB4}}$
makes the absolute value of its second derivative smaller than that of $f_k^{2,\mathrm{IASB1}}$, 
resulting in a tighter approximation.
\end{IEEEproof}

We clearly see a trade-off between approximation tightness and computational
cost (i.e. polynomial degree) for the methods IASB1, IASB3 and IASB4. 
Similarly to IASB3, the IASB4 method corresponds to an approximation
which is not necessarily convex in $s_k^n$. In Section~\ref{sec:nonconvexity}
it will be shown that this allows to escape from a
bad locally optimal solution when choosing a proper value for $q$.\\

\subsubsection{Iterative Approximation Spectrum Balancing 5
(IASB5)}\label{sec:IASB5}
The decomposition for method IASB5 is given in Table~\ref{tab:approximations}. The derivative of $f_k^{1,\mathrm{IASB5}}$
corresponds to a rational function of degree 3, resulting in a cubic equation to be solved for each
of the subproblems (\ref{eq:dualappsub}). Its tightness, after linearization, is characterized as follows:
\begin{lemma}
Approximation $f_k^{\mathrm{IASB5}}$ is tighter than $f_k^{\mathrm{IASB1}}$ and $f_k^{\mathrm{IASB4}}$.
\end{lemma}
\begin{IEEEproof}
As $w_t \log \alpha_k^t(\tilde{\mathbf{s}}_k) \log \left( \frac{s_k^t}{\mathrm{int}_k^t}\right)$ is convex in $s_k^n$, 
the absolute value of the second derivative of $f_k^{\mathrm{IASB5}}$ is 
smaller than that of $f_k^{\mathrm{IASB4}}$. Consequently, the approximation $f_k^{\mathrm{IASB5}}$ is tighter than 
that of $f_k^{\mathrm{IASB4}}$, as well as that of $f_k^{\mathrm{IASB1}}$ given Lemma~\ref{lemma:iasb4app}
\end{IEEEproof}

Similarly to IASB3 and IASB4, the IASB5 method involves an approximation
which is not necessarily convex in $s_k^n$, allowing to tackle the issue
of getting stuck in bad locally optimal solutions when choosing good
values for $q$ and $t$. 

\subsubsection{Iterative Approximation Spectrum Balancing 6 (IASB6)}
This method adds one parameterized term to $f_k^{1,\mathrm{IASB1}}$ and
$f_k^{2,\mathrm{IASB1}}$ to obtain the IASB6 decomposition as given in 
Table~\ref{tab:approximations}. $f_k^{1,\mathrm{IASB6}}$ has a derivative
equal to a rational function of degree 1, similarly to $f_k^{1,\mathrm{IASB1}}$,
resulting in a linear equation to solve each subproblem (\ref{eq:dualappsub}). More specifically,
its closed-form solution corresponds to the following:
\begin{equation}\label{eq:updateIASB6}
s_k^n = \left[ \frac{w_n(1-\beta_k^n)}{\lambda_n + d_k^n(\theta)} -
\left( \sum_{m \neq n} a_k^{n,m} s_k^m + z_k^n \right)
\right]^{s_k^{n,\mathrm{mask}}}_0.
\end{equation}
Note that the parameter $\beta_k^n$ has no impact on the computational cost
of the closed-form solution (\ref{eq:updateIASB6}). However $\beta_k^n$ has an
impact on the approximation error through $f_k^{2,\mathrm{IASB6}}$. As proposed
in (\ref{eq:tuning}) we can tune $\theta = \beta_k^n$ so as to minimize the approximation
error. For this we define the following concrete optimization problem:
\begin{equation}\label{eq:tuningIASB6}
 \beta_k^{*,n} := \argmin_{\beta_k^n} \left| \frac{\partial^2}{\partial (s_k^n)^2} 
 \left( f_k^{2,\mathrm{IASB6}} \right)\right| {~\mathrm{s.t.}~} \frac{\partial^2}{\partial (s_k^n)^2} 
 \left( f_k^{2,\mathrm{IASB6}} \right) \leq 0, \forall s_k^n \in \mathcal{S}_k^n,
\end{equation}
which minimizes the absolute value of the second derivative of $f_k^{2,\mathrm{IASB6}}$
while remaining concave so as to satisfy the design framework constraint, i.e., 
$f_k^{2,\mathrm{IASB6}} \in \mathcal{F}^2$. As we target closed-form solutions for our methods
we derive the following inequality:
\begin{equation}\label{eq:inequalityIASB6}
 \beta_k^n \leq \sum_{m \in \mathcal{N}\backslash n} \underbrace{(\mathrm{rec}_k^n)^2}_{(A)} 
  \underbrace{\frac{w_m s_k^m (a_k^{m,n})^2 }{w_n}
  \frac{1}{\mathrm{rec}_k^m \mathrm{int}_k^m}
  \left( \frac{1}{\mathrm{rec}_k^m} + \frac{1}{\mathrm{int}_k^m}\right)}_{(B)}, \forall s_k^n \in \mathcal{S}_k^n,
\end{equation}
which is obtained by taking the second derivative of $f_k^{2,\mathrm{IASB6}}$,
restricting it to be negative and extracting $\beta_k^n$ to the left side of the
inequality sign. A closed-form value for $\beta_k^n$ that satisfies (\ref{eq:inequalityIASB6})
can then be found by taking value $s_k^n=0$ for part (A) and
$s_k^n=s_k^{n,\mathrm{mask}}$ for part (B). We note that this is not the
only feasible value nor the optimal value as given by (\ref{eq:tuningIASB6}),
but presents the advantage of being computable in closed-form.

For $\beta_k^n=0$, IASB6 corresponds to IASB1. For any value $\beta_k^n > 0$, IASB6 
improves its tightness compared to that of IASB1. As $\beta_k^n$ cannot be
negative the resulting convex approximation $f_k^{\mathrm{IASB6}}$ improves that of
$f_k^{\mathrm{IASB1}}$, with similar computational cost for solving the
subproblems.

%

\subsubsection{Iterative Approximation Spectrum Balancing 7 (IASB7)}
Decomposition IASB7, as given in Table~\ref{tab:approximations}, has a 
$f_k^{1,\mathrm{IASB7}}$ whose derivative is a rational function of degree
3, restricting the computational cost of solving each of the subproblems
(\ref{eq:dualappsub}) to that of solving a cubic equation. Thus
$f_k^{1,\mathrm{IASB7}} \in \mathcal{F}^1$. $f_k^{2,\mathrm{IASB7}}$ can be
chosen to be concave by tuning $\beta_k^n$ as follows:
\begin{equation}\label{eq:inequalityIASB7}
 \beta_k^n = \sum_{m \in \mathcal{N}\backslash \{n,q\}} \underbrace{(\mathrm{rec}_k^n)^2}_{(A)} 
  \underbrace{\frac{w_m s_k^m (a_k^{m,n})^2 }{w_n}
  \frac{1}{\mathrm{rec}_k^m \mathrm{int}_k^m}
  \left( \frac{1}{\mathrm{rec}_k^m} + \frac{1}{\mathrm{int}_k^m}\right)}_{(B)},
\end{equation}
with $s_k^n=0$ for part (A) and $s_k^n=s_k^{n,\mathrm{mask}}$ for part (B).
As $\beta_k^n \geq 0$ the resulting (non-convex) approximating function $f_k^{\mathrm{IASB7}}$
improves that of IASB3.

\subsubsection{Iterative Approximation Spectrum Balancing 8 (IASB8)}
Decomposition IASB8 from Table~\ref{tab:approximations} adds a parameterized
term to decomposition IASB5, which can be tuned to improve the tightness
of the approximation $f_k^{\mathrm{IASB8}}$. $f_k^{1,\mathrm{IASB8}}$ has
a derivative that is a rational function of degree 3, resulting in a cubic
equation to solve each subproblem in (\ref{eq:dualappsub}), thus $f_k^{1,\mathrm{IASB8}} \in \mathcal{F}^1$. 
$f_k^{2,\mathrm{IASB8}}$ differs from $f_k^{2,\mathrm{IASB5}}$ in the first term, which can be tuned to linearly upper bound $f_k^{2,\mathrm{IASB8}}$ as follows:
\begin{equation}
 \beta_k^n = \min \left(0,\underbrace{\frac{(\mathrm{rec}_k^n)^2}{w_n}}_{(A)}
 \underbrace{\left( 
 \sum_{m \in \mathcal{N}\backslash n} 
 \frac{w_m s_k^m (a_k^{m,n})^2}{\mathrm{rec}_k^m \mathrm{int}_k^m} \left(\frac{1}{\mathrm{rec}_k^m}+\frac{1}{\mathrm{int}_k^m}\right)
 - \sum_{p \in \{q,t\}} \frac{w_p \alpha_k^p (a_k^{p,n})^2}{(\mathrm{int}_k^p)^2}\right)}_{(B)}\right),
\end{equation}
with $s_k^n=0$ for part (A) and $s_k^n=s_k^{n,\mathrm{mask}}$ for part (B).
The resulting approximating function $f_k^{\mathrm{IASB8}}$ is nonconvex.

\subsubsection{Iterative Approximation Spectrum Balancing 9 (IASB9)}
Decomposition IASB9, given in Table~\ref{tab:approximations}, adds a
different parameterized term to IASB5 compared to that of IASB8. $f_k^{1,\mathrm{IASB9}}$
has a derivative that is a rational function of degree 3, so that the computational cost to solve the corresponding subproblems in (\ref{eq:dualappsub}) is similar to that of IASB3, IASB5, IASB7 and IASB8. 

Tuning parameter $\theta=L_k^n$ does not impact the computational cost. It only 
impacts the tightness where the following value is chosen such that 
$f_k^{2,\mathrm{IASB9}}$ can be linearly upperbounded
\begin{displaymath}
 L_k^n=\frac{-w_q \alpha_k^q (a_k^{q,n})^2}{2 (\mathrm{int}_k^q)^2}+
 \sum_{m \in \mathcal{N} \backslash n} \frac{w_m s_k^m (a_k^{m,n})^2}{2 \mathrm{rec}_k^m \mathrm{int}_k^m} 
 \left(\frac{1}{\mathrm{rec}_k^m}+\frac{1}{\mathrm{int}_k^m} \right),
\end{displaymath}
with $s_k^n=s_k^{n,\mathrm{mask}}$. Note that this is a closed-form 
optimal solution to (\ref{eq:tuning}). The resulting approximating function
$f_k^{\mathrm{IASB9}}$ is nonconvex.

\subsubsection{Iterative Approximation Spectrum Balancing 10 (IASB10)}
We refer to our last decomposition as IASB10. We want to highlight though
that it does not fit within the proposed design framework in terms of
computational cost to solve the subproblems in (\ref{eq:dualappsub}), i.e., $f_k^{1,\mathrm{IASB10}} \notin \mathcal{F}^1$. More specifically,
it requires solving a polynomial of degree N, which is equal to the
computational cost required for solving the subproblems of existing
ICA methods (CA-DSB as well as SCALE). The improved tightness is characterized
by the following lemma:
\begin{lemma}
 Approximation $f_k^{\mathrm{IASB10}}$ is tighter than $f_k^{IASB1}$, 
 $f_k^{IASB4}$, $f_k^{IASB5}$, $f_k^{CA-DSB}$ and $f_k^{SCALE}$.
\end{lemma}
\begin{IEEEproof}
It is enough to prove that $f_k^{\mathrm{IASB10}}$ is tighter than 
$f_k^{\mathrm{IASB5}}$ and $f_k^{\mathrm{SCALE}}$, as $f_k^{\mathrm{IASB5}}$
is tighter than $f_k^{IASB1}$, $f_k^{IASB4}$, and $f_k^{CA-DSB}$.
Similarly to the proof of Lemma~{\ref{lemma:iasb4app}}, inequality
(\ref{eq:iasb4app}) can be used to prove that $f_k^{\mathrm{2,IASB10}}$
is an upper bound of $f_k^{\mathrm{2,SCALE}}$ with equality at $s_k^n=\tilde{s}_k^n$.
As a result $f_k^{\mathrm{IASB10}}$ is tighter than $f_k^{SCALE}$. 
Furthermore, $f_k^{\mathrm{2,IASB10}}$ consists of $f_k^{\mathrm{2,IASB5}}$
with some extra convex terms, making it less concave. As a result, the approximation
$f_k^{\mathrm{IASB10}}$ is tighter than $f_k^{\mathrm{IASB5}}$.
\end{IEEEproof}

\subsection{Analysis of novel methods and generalization}\label{sec:analysismethods}
Relations between the ten proposed methods of Section~\ref{sec:novelmethods} 
are important so as to understand the real trade-off between computational
cost and efficiency. Although some of these relations are already
proven and discussed in Section~\ref{sec:novelmethods}, we provide an overview
relation graph in Figure~\ref{fig:relationgraph}. The vertical axis refers
to the computational cost expressed as the polynomial degree required to
solve the subproblems (SUB) in (\ref{eq:dualappsub}) in closed-form. 
We observe some variation for each fixed degree as some approximations require the computation of additional parameters.
For instance IASB1 and IASB6 both belong to the polynomial degree 1 class,
but IASB6 requires the computation of the $\beta_k^n$ parameters and is
therefore located at a lower level (with higher computational cost) compared to that of IASB1. The horizontal
axis refers to the approximation efficiency measured as the average number
of approximations required to converge to a locally optimal solution. Only a rough ordering is
provided, based on the empirical simulation data of Section~\ref{sec:simulations}.
Arrows in the graph indicate which methods dominate others in terms of tightness. For instance, IASB2 is
dominated by IASB9 which means that the IASB9 approximation is tighter than 
the IASB2 approximation. The domination order follows successive arrows, which means
that IASB9 also dominates IASB1 and CA-DSB. Note that we only obtain a partial ordering, e.g.,
some pairs of methods (such as IASB2 and IASB5) cannot be easily compared as they involve 
different approximation approaches. One important observation is that all
proposed methods are better than existing methods CA-DSB and SCALE, both in terms
of computational cost and approximation efficiency, with the only exception
of IASB10 that has a similar computational cost.

\begin{figure}[!t]
\centering
\includegraphics[width=0.95\columnwidth]{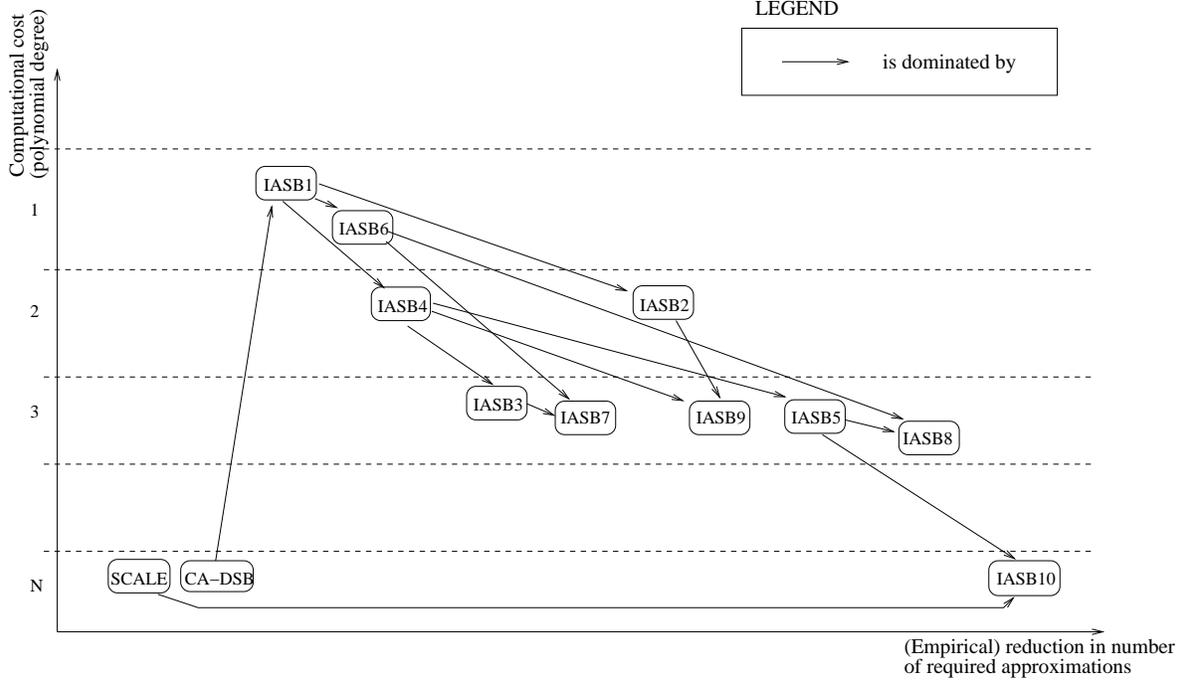}
\caption{Trade-off between computational cost (polynomial degree) and approximation efficiency (estimated empirically by the number of required approximations observed in Section~\ref{sec:simulations}, see also Table~\ref{tab:performance}).}
\label{fig:relationgraph}
\end{figure}

Although we propose a seemingly large number of different methods, we can offer some insight on their design process, and even further generalize it. More specifically, one can distinguish 
different types of approximation terms in $f_k^1$. A first type is
the concave quadratic term of $f_k^{2,\mathrm{IASB2}}$ with parameter
$L_k^n$. We will refer to this approximation term by '$L$'.
A second type corresponds to the second term of $f_k^{1,\mathrm{IASB4}}$, which
is based on inequality (\ref{eq:iasb4app}). We will refer to this approximation
term by '$\alpha$'. Note that IASB5 has two of such approximations terms, which
we will refer to as '$\alpha^2$'. A third type of approximation term is given by the second term
of $f_k^{1,\mathrm{IASB3}}$, which can be seen as a reference line, and will
be referred to as '$r$'. Finally, $f_k^{1,\mathrm{IASB6}}$, $f_k^{1,\mathrm{IASB7}}$ and $f_k^{1,\mathrm{IASB8}}$ use a '$\beta$' approximation term.
Based on these four types of approximation terms, we can rename our ten proposed
approximation methods as follows, from IASB1 to IASB10 respectively:
IA1, IA2-L, IA3-$r$, IA2-$\alpha$, IA3-$\alpha^2$, IA1-$\beta$, IA3-$\beta$r,
IA3-$\beta\alpha^2$, IA3-$\alpha$L, and IA$N$-$\alpha^{N-1}$,
where \emph{IA} refers to \emph{iterative approximation}, the first number 
denotes the polynomial degree and the following letters refer to the approximation terms used. 
This naming is more descriptive and allows to see
the order of Figure~\ref{fig:relationgraph} more clearly. This also allows 
the easy construction of other iterative approximation methods,
such as for instance IA3-$\beta\alpha$L and IA$N$-$\beta\alpha^{N-1}$, which
are better than IASB9 and IASB10, respectively.

In our design framework, the construction of an approximation method can be seen as a kind of budget allocation problem, where the budget is the degree of the polynomial equation to be solved at each iteration. All methods start with a polynomial degree equal to one. 
Adding the '$\beta$' approximation term results in no increase of polynomial degree but requires an additional
parameter computation. The 'L' approximation term adds one to the degree and requires
an additional parameter computation. The 'r' approximation term adds two to the degree
but doesn't require additional parameter computations. The '$\alpha$' approximation
term adds one to the degree and requires an additional parameter computation (and, similarly '$\alpha^2$' adds two to the degree).
Given a certain budget of polynomial degree, one can allocate different
approximation terms to obtain a per-user approximation. Degrees greater than three are also possible (such as for the last method IASB10), at the cost of increased computational work to solve the polynomial equation. Note that, in principle, polynomial equations of degree four can still be solved in closed form.

\subsection{Iterative Fixed Point Update Implementation}\label{sec:iterativefixedpoint}
As mentioned in Section~\ref{sec:ICA}, CA-DSB and SCALE cannot solve their 
corresponding univariate approximations in closed-form, as this would require solving for the roots of
polynomials of degree $N$. As a result, they have to resort to iterative 
fixed point updates. Similar iterative fixed point updates
can be derived for our proposed methods. 
Following \cite{dsb,scalejournal}, 
this can be done by starting from (\ref{eq:stationaritycondition}) and isolating the occurrence of $s_k^n$ in the first term of $f_k^{1,\mathrm{app}}$ in Table~\ref{tab:approximations}. Moving it to the left side of the equation sign leads to a fixed point update of the form $s_k^n = [h(s_k^n;\mathbf{\tilde{s}}_k^{-n}),\mathbf{\tilde{s}}_k]^{s_k^{n,\mathrm{mask}}}_0$, 
with $h$ denoting the right side function in $s_k^n$, and with projection 
on $\mathcal{S}_k^n$. For instance, these fixed point updates for IASB1 and IASB6 correspond 
to (\ref{eq:updateIASB1}) and (\ref{eq:updateIASB6}), respectively.
The drawbacks of these fixed point updates  
compared to the closed-form approach are (i) that one does not know how many 
updates are required to converge, (ii) that convergence to a global, or even to a local optimum 
of the univariate approximation is not always guaranteed, and (iii) that 
the fixed point update approach for IASB3, IASB4, IASB7, IASB8 and IASB9 
loses the beneficial ability to get out of bad locally optimal solutions,
as will be demonstrated in more detail in Section~\ref{sec:nonconvexity}. 
Nevertheless, the performance of the fixed point update approach for the methods proposed in this paper will also be assessed in the simulation Section~\ref{sec:simulations}.

\subsection{Tackling nonconvexity with nonconvex univariate approximations}\label{sec:nonconvexity}

One important advantage of using nonconvex univariate approximations is
that the issue of getting stuck in a bad locally optimal solution can
be tackled up to a certain degree. This concept is illustrated in Figure~\ref{fig:nonconvexApproximations},
where the original univariate nonconvex function, a univariate convex and 
a nonconvex approximation are plotted. When starting from the given
approximation point, one can see that the iterative approximation approach
with the convex approximation will result in convergence to the local 
minimum. Whereas, when using the nonconvex approximation with closed-form
solution, one will converge instead to the global optimum with a significantly
better objective function value. This somehow tackles nonconvexity of the original SO problem. We emphasize however that although this increases the probability to converge to the global optimal solution, it does not give a certificate for global optimality. The improved performance of some of the proposed methods using nonconvex approximations, such as IASB3 and IASB5, is demonstrated
in Section~\ref{sec:quality}.

\begin{figure}[ht]
\centering
\scalebox{0.65}{\input{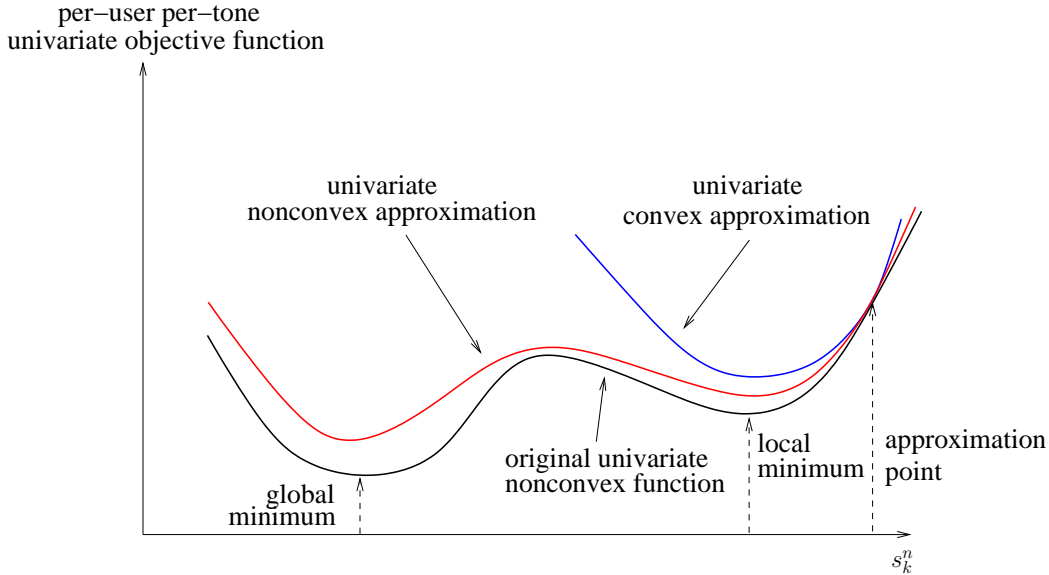}}
\caption{Tackling nonconvexity with a nonconvex univariate approximation.}
\label{fig:nonconvexApproximations}
\end{figure}

\subsection{Frequency and user selective allocation of approximations}\label{sec:allocationapproximations}

One additional degree of freedom of the proposed novel class of methods, i.e., 
Algorithm~\ref{algo:novel}, is that different approximations can be 
chosen for different tones and users. Choosing the same type of univariate approximation for each tone and user is not required as for existing ICA methods. A careful allocation of the type
of approximation over the different users and tones can in fact result
in improved solution quality with a reduced computational cost. A concrete
example of heterogeneous allocation of univariate approximation functions will be elaborated for the
DSL specific case in Section~\ref{sec:optimalallocation}, where some allocation heuristics
are derived.

\section{Simulations}\label{sec:simulations}
In this section the performance of the proposed iterative approximation 
methods will be compared to that of existing ICA methods in terms of 
convergence speed, computational cost and solution quality. For our simulations
we use a realistic DSL simulator, validated by practice and
aligned with DSL standards \cite{dsm,ETSI:ts101}. We consider 24~AWG twisted copper
pair lines. Maximum transmit power is 20.4 dBm for the ADSL and ADSL2+
scenarios, and 11.5 dBm for the VDSL scenarios. SNR gap is chosen at
12.9 dB, corresponding to a coding gain of 3~dB, a noise margin of 6~dB, and
a target symbol error probability of $10^{-7}$. Tone spacing is 4.3125 kHz. 
DMT symbol rate is 4 kHz. The weights $w_n$ are chosen equal for all users
($n=1\ldots N$), namely $w_n=1/N$, unless specified otherwise.

\subsection{Improved Convergence Speed}
To properly evaluate the performance of the methods under different settings 
we consider 10 different DSL scenarios, each of them consisting of
a different loop topology and DSL flavour. These scenarios are summarized 
in Table~\ref{tab:scenarios}. Here, 'CO-distances' denotes the distance of the lines from
the central office (CO), which is zero when the lines start from the central office (otherwise they start from a remote terminal). `Line lengths' denotes the lengths
of each of the lines, and 'DSL flavour' denotes the used DSL flavour 
as well as the transmission direction (upstream (US) or downstream (DS)).
As some methods also require the usage of reference lines, we indicate
these in the column 'Reference Lines'. 
The first number denotes the first reference line, also
denoted with index $q$, and the second number denotes the second reference line,
also denoted with the index $t$. However, when the considered
user $n$ during a given iteration is the same as one of the reference lines, i.e., 
$n=q$ or $n=t$, that reference line is changed to that corresponding to the third number indicated.
The performance of all proposed approximations is validated against optimal solutions computed with high accuracy ($0.1$~dBm) using one-dimensional exhaustive searches. These extensive simulations took several weeks to complete, which restricts us from considering even more simulation scenarios. However we want to highlight that the considered
DSL scenarios reflect more than 43000 different per-user per-tone approximations, 
and can thus be considered to give a good averaged idea of the performance
improvement of the proposed methods for general DSL scenarios.

\begin{table}
\caption{DSL scenarios}
 \label{tab:scenarios}
\begin{center}
\begin{tabular}{ | c | c | c | c | c |}
  \hline
  Scenario & Line lengths [m] & CO-distances [m] & DSL flavour & Reference Lines\\
  \hline 
  1 & 5000,4600,4200,3800,3400,3000, & 0,0,0,0,0,0, & 12-user DS ADSL2+ & [5 6 7]\\
  & 2600,2200,1800,1400,1000,600 & 1000,1000,1000,1000,1000,1000 & & \\
  \hline
  2 & 5000,4600,4200,3800,3400,3000, & 0,0,0,0,0,0, & 12-user DS ADSL2+ & [5 6 7]\\
  & 2600,2200,1800,1400,1000,600 & 0,0,0,0,0,0 & &\\
  \hline
  3 & 5000,4000,3000,2000,2000,1000, & 0,0,500,500,1000,1000, & 12-user DS ADSL2+ & [5 6 7]\\
  & 4800,3800,2800,2300,1500,1300 & 0,0,600,600,1200,1200 & &\\
  \hline
  4 & 5000,4000,3000,2000,2000,1000, & 0,0,1000,1000,2000,2000, & 12-user DS ADSL2+ & [5 6 7]\\
  & 4800,3800,2800,2300,1500,1300 & 0,0,1200,1200,2400,2400 & &\\
  \hline
  5 & 1200,1000,800,600,450,300& 0,0,0,0,0,0 & 6-user US VDSL& [1 2 3]\\
  \hline
  6 & 3000,3000,3000,3000,3000,3000,3000 & 0,0,0,0,0,0,0 & 7-user DS ADSL& [5 6 7]\\
  \hline
  7 & 5000,4000,3500,3000,3000,2500,3000 & 0,0,500,500,3000,3000,3000 & 7-user DS ADSL& [5 6 7]\\
  \hline
  8 & 5000,3000 & 0,3000 & 2-user DS ADSL& [1 2 3]\\
  \hline
  9 & 1200,600,600,600& 0,0,0,0 & 4-user US VDSL & [1 2 3]\\
  \hline
  10 & 1200,900,600,300,300,300 & 0,0,0,0,0,0 & 6-user US VDSL& [1 2 3]\\
  \hline
\end{tabular}
\end{center}
\end{table}

We started all methods from the same all-zero starting point. Since some methods may converge to different solutions, we only report in this section runs where all methods converged to the same per-user per-tone optimum, to obtain a fair comparison. Variations in the solution quality will be evaluated separately in Section~\ref{sec:quality}.

A first result of this extensive simulation setup is given in Table~\ref{tab:performance}.
For all existing and novel methods, the second column gives the average number of per-user per-tone approximations (averaged over all DSL scenarios of Table~\ref{tab:scenarios}) 
to obtain convergence within $0.1$~dBm accuracy, compared to the solution computed by exhaustive search.
%
One can see that the proposed methods require much fewer approximations to
converge compared to the existing SCALE method. Convergence is also faster 
compared to the existing CA-DSB method. In particular, methods IASB5, IASB8 and
IASB10 that use two reference lines have a significantly improved convergence
speed, e.g., IASB10 requires twice as fewer approximations compared to CA-DSB, and four times fewer
approximations compared to SCALE, while requiring the same computational cost
for closed-form solution as was highlighted in Table~\ref{tab:approximations}.
It can also be seen that IASB2 is better than IASB1 as predicted by the 
theory. Similarly IASB6 is better than IASB1, IASB7 is better than IASB3, 
IASB8 is better than IASB5, and IASB9 is better than IASB4. In general
the average number of approximations varies between 1 and 2.5 for the novel
methods. 

To highlight the variance of the number of approximations we refer to 
Figure~\ref{fig:CDFapproximations}, where the cumulative distribution function (CDF)
is given for the number of approximations (averaged over all DSL scenarios 
of Table~\ref{tab:scenarios}). It can be seen that the proposed methods
have a strictly better performance compared to existing methods. Their number of appproximations varies between 1 and 6, in contrast to SCALE for which sometimes even up to 20 approximations are necessary for convergence.
Especially the IASB10 method has a very good performance, requiring at most
two approximations to converge; IASB5 and IASB8 also perform well.
One interesting observation is that on average $62.5\%$ of all cases only require one
approximation to converge. For the remaining $37.5\%$, the number of required
approximations varies. In particular,  we see for SCALE that if the first iteration is not enough, at least four or more iterations are required. 
Furthermore, one can see that the
differences between the methods are not extremely large. This is actually because of the averaging over multiple DSL scenarios.
In Figure~\ref{fig:approximationsADSLJournal} we show for instance the CDF of the 
number of approximations averaged only for scenario 7 of Table~\ref{tab:scenarios}.
Here much larger differences can be observed between the methods.
One can also see here that IASB10 and the methods using multiple reference
lines IASB5 and IASB8 perform best. This comes at the cost of solving
a polynomial of degree N, degree 3 and degree 3, respectively. The methods
IASB1 and IASB6, which only require solving a polynomial of degree 1, also
perform quite well, and can thus be seen as very efficient methods.


\begin{table}
\caption{Average number of approximations and fixed point update
iterations for different iterative approximation methods averaged over
all scenarios of Table~\ref{tab:scenarios}.}
\label{tab:performance}
\begin{center}
\begin{tabular}{ | c | c | c | }
  \hline
   Method & Avg number of  & Avg number of fixed \\
   & per-user approximations & point update iterations\\
   \hline 
   CA-DSB & 2.626 & 2.816\\
   \hline 
   SCALE & 5.630 & 6.581\\
   \hline 
   IASB1 & 2.461 & 2.461\\
   \hline 
   IASB2 & 2.414 & 2.423\\
   \hline 
   IASB3 & 2.419 & 2.461\\
   \hline 
   IASB4 & 2.421 & 2.457\\
   \hline 
   IASB5 & 2.192 & 2.480\\
   \hline
   IASB6 & 2.459 & 2.459\\
   \hline
   IASB7 & 2.416 & 2.458\\
   \hline
   IASB8 & 2.190 & 2.478\\
   \hline
   IASB9 & 2.391 & 2.432\\
   \hline
   IASB10 & 1.262 & 1.999\\
   \hline
\end{tabular}
\end{center}
\end{table}

\begin{figure}[!t]
\centering
\includegraphics[width=0.65\columnwidth]{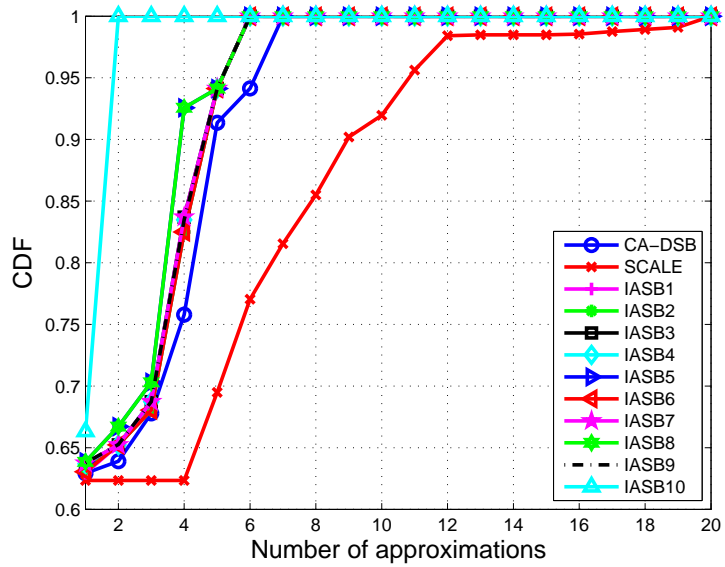}
\caption{Cumulative distribution function (CDF) of the average number of approximations required to converge
(averaged over all scenarios of Table~\ref{tab:scenarios}).}
\label{fig:CDFapproximations}
\end{figure}

\begin{figure}[!t]
\centering
\includegraphics[width=0.65\columnwidth]{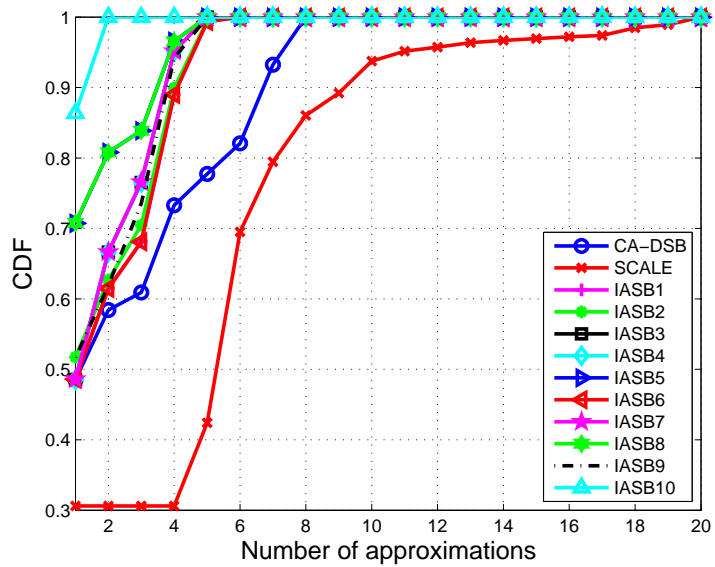}
\caption{Cumulative distribution function (CDF) of the average number of approximations required to converge
for scenario 7 of Table~\ref{tab:scenarios}}
\label{fig:approximationsADSLJournal}
\end{figure}

\begin{figure}[!t]
\centering
\includegraphics[width=0.65\columnwidth]{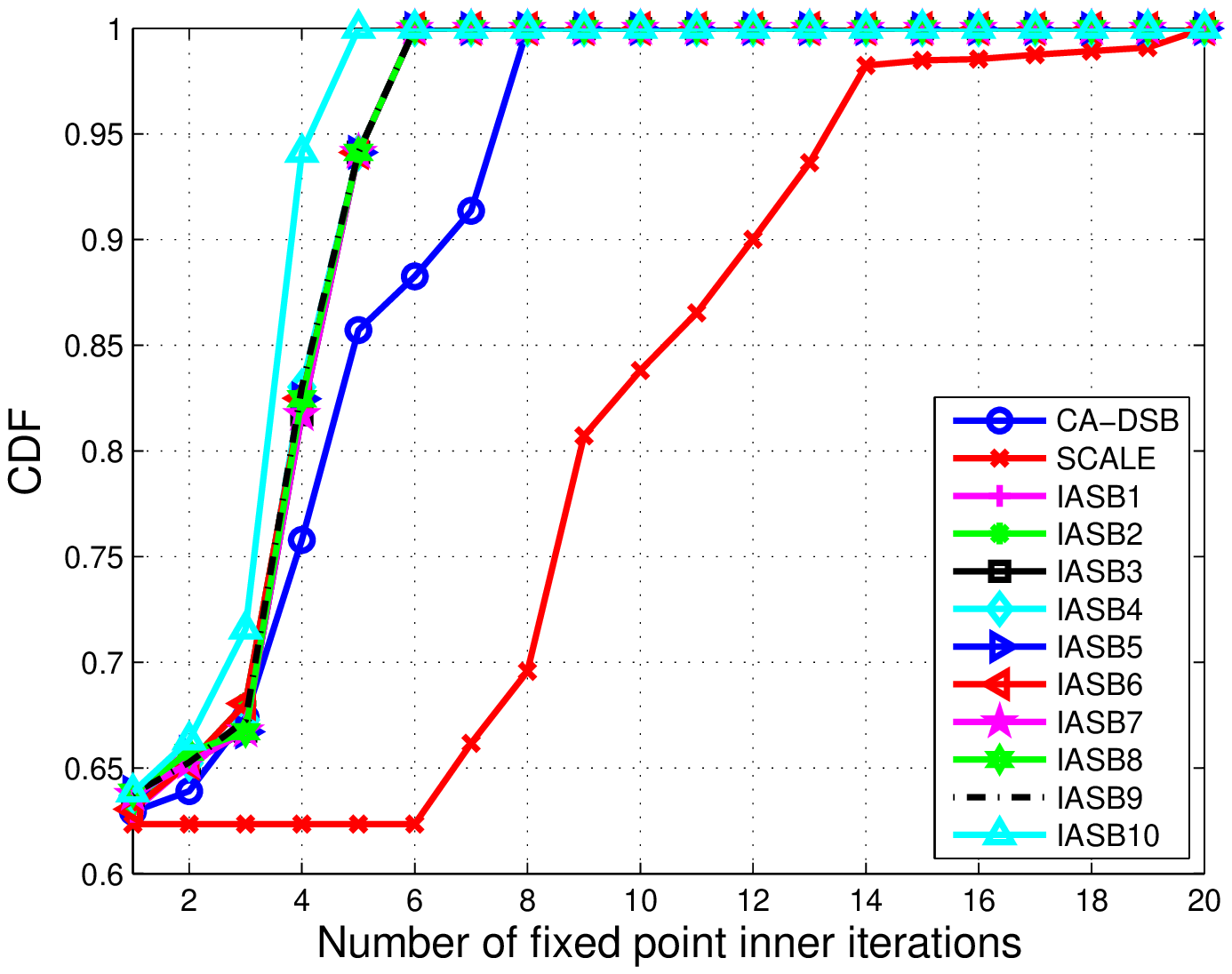}
\caption{Cumulative distribution function (CDF) of the average number of fixed point updates required to converge
(averaged over all scenarios of Table~\ref{tab:scenarios}).}
\label{fig:CDFiterations}
\end{figure}

The third column of Table~\ref{tab:performance} displays results
for which the univariate approximations are solved up to accuracy $0.1~$dBm using a fixed point update approach (instead of using the closed-form formula). 
It can be seen that even when using the fixed point update approach, as highlighted in Section~\ref{sec:iterativefixedpoint},
the newly proposed methods show a significant improvement in convergence speed.
All proposed methods perform better than existing methods, in particular IASB10.
The CDF of the average number of fixed point updates (averaged over all
DSL scenarios of Table~\ref{tab:scenarios}) is given in Figure~\ref{fig:CDFiterations}
to highlight the variance between the methods. The same observations can be made
as for the average number of approximations, except that the benefit of
IASB10 compared to the other methods is no longer so large.

Finally, we would like to highlight that the ISB method in \cite{ISB_Raphael}
is also an iterative per-user algorithm in which each subproblem is solved using an 
exhaustive search. The ISB algorithm however has two considerable disadvantages
compared to the proposed methods:
(i) solving a per-user subproblem with a simple exhaustive search requires much more computational effort, e.g., it takes several minutes of Matlab execution time to find the solution with ISB whereas only a few seconds are needed for the proposed methods. (ii) ISB considers a finite discrete granularity in transmit powers in contrast to the proposed methods, which negatively impacts the convergence speed. 

\subsection{Improved Solution Quality}\label{sec:quality}
In this section we focus on the improved solution quality, i.e., 
improved data rate performance, when using the proposed
iterative approximation methods with non-convex univariate approximations
such as IASB3.

First we will demonstrate this starting from scenario 5 in Table~\ref{tab:scenarios},
which is a 6-user US VDSL scenario, and focusing on one particular
(per-user per-tone) univariate approximation. In Figure~\ref{fig:approximations}
we plot the univariate approximations for user 5 on tone 600 for different
methods, where we consider the per-user weighted achievable bit loading
which corresponds to $-f_{600}^{\mathrm{app}}(s_{600}^5;\mathbf{\tilde{s}}_{600}^{-5},\mathbf{\tilde{s}}_{600},\theta)$. The approximation point is chosen at the spectral mask -30 dBm. 
It can be seen that the CA-DSB, SCALE, IASB1 and IASB2 approximations get 
stuck in a bad locally optimal solution at -30 dBm when solved using the 
closed-form approach, whereas the nonconvex approximations of IASB3, IASB4 
and IASB5 succeed in getting out of the bad locally optimal solution 
and result in the true global optimum at 0~dBm. In fact, IASB3 matches 
the original nonconvex objective. Note that we omit the plots of the 
IASB6, IASB7, IASB8, IASB9 and IASB10 approximations to avoid an overcrowded
figure. 

A second example focuses on scenario 8 of Table~\ref{tab:scenarios}, which
is an important near-far 2-user DS ADSL setting for which it is known that it
is difficult to find a good solution. In Table~\ref{tab:nonconvexcomparison}, we show
the final achievable data rate performance of the IASB1 and IASB3 methods in terms of weighted 
achievable data rate sum as well as individual achievable data rates. Note that we choose non-equal
weights for this scenario to prevent that the CO line of 5000m has a 
too small data rate performance, i.e., $w_1=1.2984$ and $w_2=0.1443$ which
are normalized with $\log(2)$.  
Although the weighted achievable data rate
sum is only slightly larger for IASB3, one can see that the individual achievable data
rate $R^2$ of the RT-user is $7\%$ better while the CO-user's achievable data rate $R^1$
is nearly unchanged. The reason for this improved solution 
quality is that the nonconvex approximations of IASB3 succeed in getting
out of bad locally optimal solutions, resulting in an overall better 
achievable data rate performance. In Figure~\ref{fig:transmitSpectra} the resulting
transmit spectra are shown for both the CO- and RT-users and for both methods IASB1
and IASB3. A clear difference can be seen in the tone range
110 to 123, where IASB3 succeeds in obtaining
better solutions with its nonconvex approximations, resulting 
in an overall improved solution quality.

\begin{table}
\caption{Comparison of solution quality for scenario 8 of Table~\ref{tab:scenarios}
with IASB1 (using convex approximations) and IASB3 (using nonconvex approximations).}
\label{tab:nonconvexcomparison}
\begin{center}
\begin{tabular}{ | c | c | c | }
  \hline
   Performance measure & IASB1  & IASB3 \\
   \hline
   $w_1 R^1 + w_2 R^2$ [Mbit/s] & $2.73~10^6$ & $2.78~10^6$ \\
   \hline
   $R^1$ [Mbit/s] & $1.60~10^6$ & $1.59~10^6$\\
   \hline
   $R^2$ [Mbit/s] & $\mathbf{4.55~10^6}$ & $\mathbf{4.87~10^6}$ \\
   \hline
\end{tabular}
\end{center}
\end{table}

We also want to highlight here that the improvement in solution quality obtained from the use of
nonconvex approximations is only observed when using the closed-form solution approach, and is
not obtained for the iterative fixed point update approach. For instance, for Figure~\ref{fig:approximations}, 
all methods would get stuck in -30 dBm when using the iterative fixed point
update approach. This is an important advantage supporting the use of the closed-form 
solution approach with the proposed methods. Note however that for CA-DSB, SCALE and IASB10 the
closed-form solution approach can not be used when considering scenarios with more than four users, as this requires computing the roots of some polynomials of degree larger than four, which is not always straightforward or even possible.
\begin{figure}[!t]
\centering
\includegraphics[width=0.65\columnwidth]{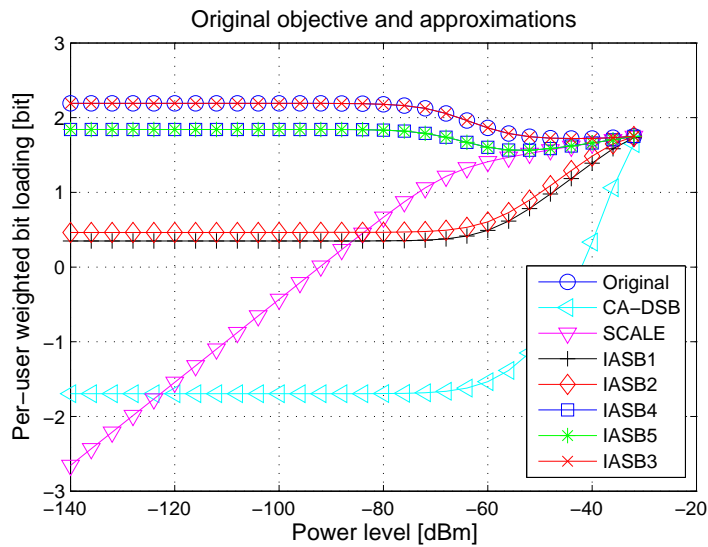}
\caption{Original nonconvex per-user objective and its approximations}
\label{fig:approximations}
\end{figure}

\begin{figure}[!t]
\centering
\includegraphics[width=0.75\columnwidth]{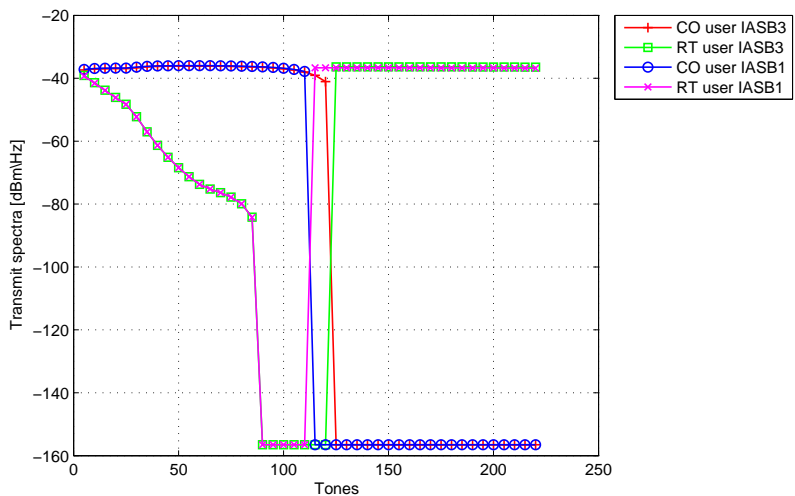}
\caption{Transmit spectra scenario 8 from Table~\ref{tab:scenarios} using IASB1 and IASB3.}
\label{fig:transmitSpectra}
\end{figure}

\subsection{Per-user Per-tone Selective Allocation of Approximations}\label{sec:optimalallocation}
As mentioned in Section~\ref{sec:allocationapproximations}, the proposed novel class of methods
allows to exploit an extra degree of freedom that consists in allocating
different approximations to different users and tones. For the per-user
per-tone problems that are easy to approximate and convex, choosing a low-complexity approximation such as IASB1 and IASB6 is sufficient. Then again, 
other per-user per-tone problems may have a very nonconvex behaviour, with
multiple local optima, for which it is more interesting to use nonconvex
approximations, at the cost of solving a cubic equation instead of a simple linear
equation.

Although it is not easy to determine which per-user per-tone problems
are 'easy', we can try to use some intuition. DSL scenarios
can consist of mixed settings with CO-users and RT-users. These mixed
settings are known to be scenarios with asymmetric crosstalk, i.e., the
RT-users cause a lot of crosstalk to CO-users whereas the CO-users cause
only minor crosstalk impact on the RT-users. From the univariate objective
point of view, this translates to having simple univariate convex objective
functions for the CO-users whereas having univariate non-convex objective
functions for the RT-users. The connection between a small crosstalk level
and convexity was also formally proven in \cite{Tsiaflakis2008a}. 
For the CO-users the usage of convex approximations
is good enough to converge to the univariate local optima, whereas for the
RT-users one can benefit from nonconvex approximations to obtain an improved
solution quality. We have performed simulations for the CO-RT scenario 8 of Table~\ref{tab:scenarios} where the IASB1 approximation is used for
the CO-line of 5000m and the IASB3 approximation is used for the RT-line of 
3000m. The resulting hybrid approach requires less computational
effort compared to the IASB3 method, as the CO-user only has to solve polynomials
of degree 1 instead of degree 3. The performance of the hybrid method is
the same as that of the IASB3 method, i.e., superior to that of those based on 
convex approximations. One can thus have the same improved 
performance of IASB3 with a significantly decreased computational cost.

Similarly one can also make the approximations tone-dependent, i.e., frequency selective. 
For instance, the crosstalk is larger at high frequencies and one can expect that non-convex
approximations may lead to better performance, in contrast to the low frequency
bands with small crosstalk for which low complexity convex approximations are 
good enough. An optimized allocation of approximation types over different tones
and users is an interesting topic for future research, but out of the 
scope of this paper.

\subsection{Choice of methods}
Although a thorough trade-off analysis (theoretically as well as empirically) 
is conducted for the different proposed methods, such as the relation 
graph of Figure~\ref{fig:relationgraph}, it is not easy to give concrete
recommendations on which method is the best. This is because this depends 
on the relative practical (hardware) cost of computing the solution of a
polynomial of degree 1, 2 and 3. Furthermore the best choice depends on
how much importance is given on preventing the method from getting stuck 
in bad locally optimal solutions.
If solution quality is extremely important one should go for one of the
nonconvex approximations, or a hybrid combination that takes the specific
scenario into account as highlighted in Section~\ref{sec:optimalallocation}.

\section{Conclusion}
Multi-user multi-carrier systems that follow an interference channel based
system model have become very important in practice. In such systems, 
the use of spectrum optimization can tackle the interference problem so 
as to spectacularly increase the achievable data rates. However this requires the development
of efficient optimization methods to solve the corresponding nonconvex
optimization problems. In this paper we demonstrate that existing spectrum optimization
methods can be significantly improved when taking their typical per-user implementation
explicitly into account at the early design stage. Applying this concept,
we derived a novel class of iterative approximation methods which focuses
on efficiently solvable univariate approximations, instead of the jointly convex approximations
of existing methods (CA-DSB and SCALE). For this novel class, a design
framework is proposed that can be used to construct improved approximations
that are much tighter than existing approximations and that can be solved 
in closed-form, in contrast to existing methods that require iterative
fixed point updates. As a result, methods belonging to this novel class both converge in fewer iterations
and require fewer computations per iteration. Using this design 
framework, we construct ten novel methods referred to as IASB1 up to IASB10,
and analyze and discuss their improved tightness and reduced computational cost. 
Methods IASB1, IASB2 and IASB6 involve convex univariate approximations, whereas
the other methods rely on nonconvex univariate approximations. 
Extensive simulations using a realistic multi-user DSL simulator demonstrate
a factor 2 to 4 improvement in convergence speed, while decreasing 
the computational cost per iteration. Furthermore the methods using
nonconvex approximations somehow allow to tackle nonconvexity of the optimization problem: since 
they allow escape from bad locally optimal solutions, they can improve solution 
quality compared to existing methods, which is demonstrated for IASB3, IASB4 and IASB5.
In particular, methods using multiple reference lines, namely IASB5, IASB8 as well
as IASB10, perform very well in terms of their trade-off between computational
cost, convergence speed and solution quality. Finally it is shown
that the proposed algorithm permits the selection of different approximations over
different users and tones, which allows for further performance improvements,
as demonstrated for a mixed DSL setting.

\bibliographystyle{IEEEtran}
\bibliography{bibliographyshort}

\end{document}